\RequirePackage{amsmath} 
\documentclass{llncs}
\usepackage{amsfonts,amsmath,amssymb,bm,verbatim}
\usepackage{braket,chngcntr,color,colortbl,epsfig,graphicx,xr-hyper,hyperref,multirow,tabularx,tikz}
\usepackage[misc]{ifsym}
\usepackage[all]{xy}

\newcommand{\nc}{\newcommand}

\nc{\keywords}[1]{\par\addvspace\baselineskip\noindent\keywordname\enspace\ignorespaces#1}

\nc*\circled[1]{\tikz[baseline=(char.base)]{\node[shape=circle,draw,inner sep=0.01pt] (char) {#1};}}
\nc{\nn}{\nonumber \\}
\nc{\ds}{\displaystyle}
\nc{\sgn}{\rm sgn}
\nc{\T}{{\rule{0pt}{2.2ex}}}
\nc{\B}{{\rule[-1.0ex]{0pt}{0pt}}}

\nc{\kett}[1]{{|{#1}\rangle}}
\nc{\braa}[1]{{\langle{#1}|}}
\nc{\braakett}[2]{{\langle{#1}|{#2}\rangle}}
\nc{\commentout}[1]{}
\nc{\mxp}{{\rm mp}}
\nc{\rnd}{{\rm rand}}
\nc{\ks}{{\rm KS}}
\nc{\pis}{{\rm PS}}
\nc{\cf}{{\rm CF}}
\nc{\cfc}{{\rm CNS}}
\nc{\cfg}{{\rm GwDP}}
\nc{\ip}{{\rm IP}}
\nc{\op}{{\rm OP}}
\nc{\aset}{{\{0,1\}^*}}
\nc{\bs}{{\rm BP}}
\nc{\rs}{{\overset{\$}{\leftarrow}}}
\nc{\acct}[1]{AES-C{#1}}
\nc{\scct}[1]{SHA-C{#1}}
\nc{\from}{\colon}
\nc{\eqnsp}{ }

\renewcommand{\qed}{\hfill\square}

\nc{\mapstr}[1]{{ \multirow{3}{*}{$\overset{\textrm{#1}}{\longmapsto}$} }}
\nc{\mapdstr}[3]{{ \multirow{3}{*}[#3em]{$\overset{\substack{\textrm{#1}\\\textrm{#2}}}{\longmapsto}$} }}

\begin{document}

\title{Time-Space Complexity of Quantum Search Algorithms in Symmetric Cryptanalysis}

\author{Panjin Kim \and Kyung Chul Jeong$^{(\text{\Letter})}$ \and Daewan Han}
\institute{National Security Research Institute, Daejeon 34044, Korea \\
jeongkc@nsr.re.kr}

\maketitle
\begin{abstract}
Performance of cryptanalytic quantum search algorithms is mainly inferred from \emph{query} complexity which hides overhead induced by an implementation.
To shed light on quantitative complexity analysis removing hidden factors, we provide a framework for estimating
time-space complexity, with carefully accounting for characteristics of target cryptographic functions.
Processor and circuit parallelization methods are taken into account, resulting in the time-space trade-offs curves in terms of \emph{depth} and \emph{qubit}.
The method guides how to rank different circuit designs in order of their efficiency.
The framework is applied to representative cryptosystems NIST referred to as a guideline for security parameters, reassessing the security strengths of AES and SHA-2.
\end{abstract}

\section{Introduction}\label{sec:intro}


Quantum cryptanalysis is an area of study that has long been developed alongside the field of quantum computing, as many cryptosystems are expected to be directly affected by quantum algorithms. One of the quantum algorithms that would have an impact on symmetric cryptosystems is Grover's algorithm~\cite{grover97}.
It had been widely known that many symmetric cryptosystem's security levels will be simply reduced by half due to the asymptotic behavior of the query complexity of Grover's algorithm under the oracle assumption.
As the field has matured over decades, not mere asymptotic but more quantitative approaches to the cryptanalysis are also being considered recently~\cite{sha16,aes16,revs15,ecdlp17,mq16}.
These works have substantially improved the understanding of quantum attacks by systematically estimating quantum resources.
Nevertheless, it is still noticeable that the existing works on resource estimates are more intended for suggesting exemplary quantum circuits (so that one can count the number of required gates and qubits explicitly) than fine-tuning of actual attack designs.

The importance of estimating costs of quantum search algorithms beyond pioneering works should be emphasized as it can be utilized to suggest practical security levels in the post-quantum era.
NIST indeed suggested security levels based on the resistances of AES and SHA to quantum attacks in the PQC standardization call for proposals document~\cite{nist-quantum}. In addition, the difficulty of measuring the complexity of quantum attacks was questioned in the first NIST PQC standardization workshop \cite{moody18}.
The main purpose of this work is to formulate the time-space complexity of quantum search algorithm in order to provide reliable quantum security strengths of classical symmetric cryptosystems.

\subsection{This work}
There exist two noteworthy points overlooked in the previous works.
First, the target function to be inverted
is generally a pseudo-random function or a cryptographic hash function. 
Under the characteristics of such functions, bijective correspondence between input and output is not guaranteed.
This makes Grover's algorithm \emph{seemingly} inadequate due to the unpredictability of the number of targets.
The second point is time-space trade-offs of quantum resources.
Earlier works on quantitative resource estimates have implicitly or explicitly assumed a single quantum processor.
Presuming that the \textit{resource} in classical estimates includes the number of processors the adversary is equipped with, the single processor assumption is something that should be revised.

Being aware of the issues, we come up with a framework for analyzing the time-space complexity of cryptanalytic quantum search algorithms. The main consequences we presented in this paper are three folded:

\subsubsection{Precise query complexity involving parallelization.}

The number of oracle queries, or equivalently Grover iterations, is first estimated as exactly as possible, reasonably accounting for previously overlooked points.
Random statistics of the target function are carefully handled which lead to increase in iteration number compared with the case of unique target.
Surprisingly, however, the cost of dealing with random statistics in this paper is not expensive compared with the previous work \cite{aes16} under single processor assumption.
Furthermore, when processor parallelization is considered, we observed that this extra cost gets even more negligible.
It is also interesting to recognize that the parallelization methods could vary depending on the search problems.
After taking the asymptotical big O notation off, the relation between time and space in terms of Grover iterations and number of processors, called \emph{trade-offs curve}, is obtained.
Apart from resource estimates, investigating the trade-offs curve of state-of-the-art collision algorithm in~\cite{Chailloux2017} with optimized parameters is one of our major concerns.

\subsubsection{Qubit-depth trade-offs and circuit design tuning.}
In the next stage, time and space resources are defined in a way that they can be interpreted as physical quantities.
Cost of quantum circuits for cryptanalytic algorithms can be estimated in units of \emph{logical qubits} and \emph{Toffoli-depths}.
Taking the total number of gates as time complexity disturbs accurate estimates for the speed of quantum algorithms due to far different overheads introduced by various gates in real operation.
With the definitions of quantum resources, the trade-offs curve now describes the relation between number of qubits and circuit depths.
Since we are given a `relation' between time and space, it is then possible to grade the various quantum circuits in order of efficiency.
In other words, the method described so far enables one to tell which attack design is more cost-effective.

By applying generic methodology newly introduced, time-space
complexities of AES and SHA-2 against quantum attacks are measured in the following way.
Various designs are constructed by assembling different circuit components with options such as reduced depth at the cost of the increase in qubits (or vice versa).
Design candidates are then subjected to the trade-offs relation for comparison.
The trade-offs coefficient of the most efficient design represents the hardness of quantum cryptanalysis.
Compared with pre-existing circuit designs, we have improved the circuits by reducing required qubits and/or depths in various ways.
However, we do not claim that we have found the optimal attacks for AES and SHA-2.
The method enables us to select the `best' one out of candidates at hand.
It is remarked that the explicit circuit designs for quantum collision search algorithms is first introduced.



\subsubsection{Revisiting the security levels of NIST PQC standardization.}

The procedure is applied to each primitive of security strength categories NIST specified in~\cite{nist-quantum}.
A new threshold that is required for the category classification, based on the cost metric proposed in this work, is provided in Fig.\;\ref{fig:intro_security}.
It includes wide range of parameters
and the quantum collision finding algorithms which do not outperform classical counterparts to explicitly recognize the quantum-side complexities of all the categories.

\begin{figure}[htbp]
    \centering
    \includegraphics[width=0.8\textwidth]{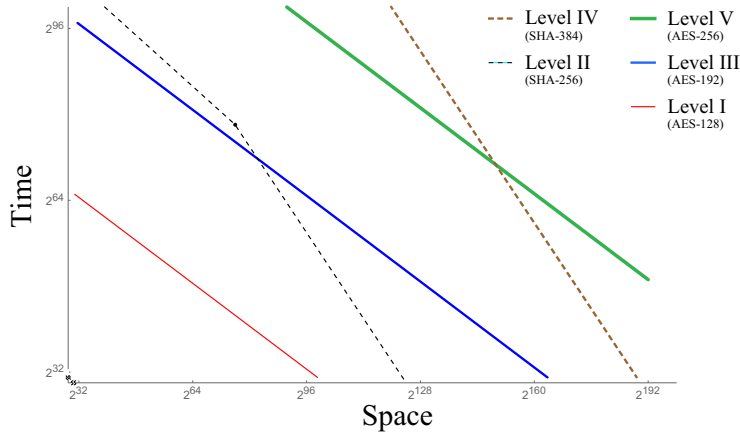}
    \caption{Time-space cost, which has conditional ordering, of the quantum attacks on five security strength category representatives of NIST PQC standardization}
    \label{fig:intro_security}
\end{figure}

We end this subsection with an important caveat.
Use of classical resources appears in this paper, but we do not handle the complexity induced by it because of unclear comparison criteria for quantum and classical resources.

\subsection{Organization}
Next section covers the backgrounds including Grover's algorithm and its variants.
In Sect.~\ref{sec:expected_iter}, the time-space complexity of relevant algorithms in a unit of Grover iteration is investigated.
A basic unit of quantum computation is proposed in Sect.~\ref{sec:quantum_cost}.
Sections~\ref{sec:AES},~\ref{sec:SHA} and~\ref{sec:SHA-coll} show the results of applying the time-space analysis to AES and SHA-2.
In Sect.~\ref{sec:sec-str}, based on the observations made in the previous sections, a comprehensive figure summarizing the quantum security strengths of AES and SHA-2 is drawn.
Section~\ref{sec:con} summarizes the paper.



\section{Backgrounds}\label{sec:back}
Grover's algorithm, the success probability, parallelization methods, and some generalizations or variants are explained briefly.
A breif review of AES and SHA-2, and an introduction to related works on resource estimates are followed.
We do not cover the basics of quantum computing, but leave the references~\cite{PQC,QCQI} for interested readers.
In this paper, every bra or ket state is normalized.

\subsection{Grover's Algorithm}
Throughout the paper, the target function is denoted by $f$ and $N=2^n$ for some $n\in\mathbb{N}$. Consider a set $X$ of size $N$ and a function $f\from X \rightarrow \{0,1\}$,
\eqnsp
\begin{align*}
   f(x) = \left\{
            \begin{array}{ll}
              1\enspace, & \hbox{~~if $x \in T$ \enspace,} \\
              0\enspace, & \hbox{~~otherwise \enspace,}
            \end{array}
          \right.
\end{align*}
where $T$ of size $t$ is a set of targets to be found.

Grover's algorithm~\cite{grover97} is an algorithm that repeatedly applies an operator $Q=-A S_0 A^{-1} S_f$, called \emph{Grover iteration}, to the initial state $\kett{\Psi} = A \kett{0}$, where
\eqnsp
\begin{align}\label{eq:grover}
     A = H^{\otimes n} \enspace,
   ~S_0 = I - 2 \kett{0}\braa{0} \enspace,
   ~S_f = I - 2 \kett{\tau}\braa{\tau} \enspace,
\end{align}
where $H^{\otimes n}$ is a set of Hadamard operators
and
$\kett{\tau}$ is a target state which is an equal-phase and equal-weight superposition of $\kett{x}$ for all $x\in T$.
The roles of $S_0$ and $S_f$ are to swap the sign of zero and $\kett{\tau}$ states, respectively.

The operators $S_f$ and $-A S_0 A^{-1}$ are known as \emph{oracle} and \emph{diffusion} operators, respectively.
By acting the oracle operator on a state, only the target state is marked through the sign change.
The diffusion operator flips amplitudes around the average.

Success probability of measurement as a function of the number of iterations has been studied in~\cite{BBHT98}, observing the optimal number of iterations that minimizes the ratio of the iterations to success rate.
We introduce the results below with notation that is used throughout the paper.

By applying $Q$ on the initial state $i$-times, the success probability of measuring one of the $t$ solutions in the domain of size $N$, denoted by $p_{t,N}(i)$, becomes
\eqnsp
\begin{align}\label{eq:suc_prob}
  p_{t,N}(i) = \sin^2 \left[ (2i+1) \theta_{t,N} \right] \enspace,
\end{align}
where $\sin \left(\theta_{t,N}\right)  = \sqrt{{t}/{N}} \left( = \braakett{\Psi}{\tau} \right)$. The number of repetitions of $Q$ maximizing the success probability of measurement, denoted by $I_{t,N}^{\mxp}$, is estimated as
\eqnsp
\begin{equation}\label{eq:iter_mp}
    I_{t,N}^{\mxp}=(\pi/4)\cdot \sqrt{N/t}  \enspace.
    \footnote{
    Rounding function is not explicitly used in this paper for simplicity.}
\end{equation}

When the measurement is made after $i$-repetitions of $Q$, the expected number of Grover iterations to find one of the targets can be expressed as a function of $i$. For $t$ targets in the domain of size $N$, the function is denoted by $I_{t,N}(i)$ which reads $I_{t,N}(i)=i/p_{t,N}(i)$.
The optimal number of iterations $i_{t,N}$ that minimizes $i / p_{t,N}(i)$ is found to be $i_{t,N}=0.583\ldots\cdot \sqrt{N/t}$, and then the expected number of iterations, denoted by $I_{t,N}$ reads
\eqnsp
\begin{align} 
  \label{eq:iter_t_val}
  I_{t,N} = I_{t,N}(i_{t,N}) = 0.690\ldots\cdot\sqrt{N/t} 
   \enspace.
\end{align}
In some cases where the domain size is $N$, it is omitted such as $p_{t}(i)(=p_{t,N}(i))$ or $I_{t}(=I_{t,N})$, for readability.

Parallelization of Grover's algorithm using multiple quantum computers has been investigated in applications to cryptanalysis~\cite{bernstein17,nist-quantum,zalka99}. Consideration of parallelization in an hybrid algorithm can be found in~\cite{BY18}. Asymptotically the execution time is reduced by a factor of the square root of the number of quantum computers. There are two straightforward parallelization methods having such property, called \emph{inner} and \emph{outer} parallelization.

Here we fix some notations for parameters related to parallelization. $T_q$ and $S_q$ stand for the number of sequential Grover iterations and the number of quantum computers, respectively. $S_c$ stands for the amount of classical resources, such as the size of storage and/or the number of processors. Definitions of two parallelization methods can be given as follows.

\begin{definition}[Inner Parallelization (IP)]\label{def:ip}
After dividing the entire search space into $S_q$ disjoint sets, each machine searches one of the $S_q$ sets for the target. The number of iterations can be reduced due to the smaller domain size.
\end{definition}

\begin{definition}[Outer Parallelization (OP)]\label{def:op}
Copies of Grover's algorithm on the entire search space is executed in $S_q$ machines. Since it is successful if any of $S_q$ machines finds the target, the number of iterations can be reduced.
\end{definition}
%


Parallelization is inevitable once the notion of MAXDEPTH~\cite{nist-quantum}
is applied.

\subsection{Generalizations and Variants}\label{subsec:variants}

Fixed-point~\cite{fixed-search} and quantum amplitude amplification (QAA)~\cite{QAA} algorithms are generalizations of Grover's algorithm. A brief review of QAA is given in this subsection which appears as a component of a collision finding algorithm in later sections.
We skip over the fixed-point algorithm as it has no advantage over Grover's algorithm and QAA in this work.\commentout{APPDX}\footnote{There are two reasons. One is that Fixed-point search requires \textit{two} oracle queries per iteration, and the other is $\log (2/\delta)$ factor in (3) in \cite{fixed-search} which also increases the required number of iterations depending on the bounding parameter $\delta$. Comparing these factors with the overhead in our method introduced by random statistics, we concluded that the fixed-point algorithm is not favored.}

There exist a number of variants of Grover's algorithm in application to collision finding. In \cite{Brassard1998}, Brassard, H{\o}yer and Tapp suggested a quantum collision finding algorithm (BHT) of $O(N^{1/3})$ query complexity using \emph{quantum memory} amounting to $O(N^{1/3})$ classical data.
A multi-collision algorithm using BHT was suggested in~\cite{Hosoyamada2017}.
In this work however, we do not consider BHT as a candidate
algorithm for the following reasons.
One is that the algorithm entails a need for quantum memory where the realization and the usage cost are controversial~\cite{qram-cost}, and the other is that we are unable to come up with any implementation restricted to use of elementary gates that do not exceed the total cost of $O(N^{1/2})$.

Apart from quantum circuits, algorithms primarily designed for other type of models such as measurement-based quantum computation also exist, for example quantum walk search~\cite{qwalk03,qwalk04} or element distinctness~\cite{ambainis07}, but we do not cover them as state-of-the-art quantum architecture is targeting for circuit computation.
Interested readers may further refer to~\cite{Hosoyamada2017} and related references therein for more information on quantum collision finding.



Bernstein analyzed quantum and classical collision finding algorithms in~\cite{bernstein-collision}. Quoting the work,
no quantum algorithm with better time-space product complexity than $O(N^{1/2})$ which is achieved by the state-of-the-art classical algorithm~\cite{Rho-DP} had not been reported. If Grover's algorithm is parallelized with the distinguished point method, complexity of $O(N^{1/2})$ can be achieved. This is one of the examples of \emph{immediate ways} to combine quantum search with the rho method as mentioned in \cite{bernstein-collision}. We denote it as Grover with distinguished point (GwDP) algorithm in this paper.

In Asiacrypt 2017, Chailloux, Naya-Plasencia and Schrottenloher suggested a new quantum collision finding algorithm, called CNS algorithm, of $O(N^{2/5})$ query complexity using $O(N^{1/5})$ classical memory~\cite{Chailloux2017}.

\subsubsection{QAA algorithm.}
Basic structure of QAA is the same as Grover's original algorithm.
Initial state $|\Psi \rangle = A |0 \rangle$ is prepared, and then Grover iteration $Q$ is repeatedly applied $i$ times to get success probability (\ref{eq:suc_prob}).
The only difference is that in QAA, the preparation operator $A$ is not restricted to $H^{\otimes n}$ where $N=2^n$, and so thus the search space can be arbitrarily defined.
Detailed derivation is not covered here, but instead we describe the key feature in an example.

As a trivial example, let us assume we are given a quantum computer, and try to find a target bit-string 110011 in a set $N=\{ x\; |\; x \in \{0,1\}^6\; \textrm{and two}$ $\textrm{middle}$ $\textrm{bits are 0} \}$. \commentout{APPDX}
Domain size is not equal to $2^6$, and the initial state can be prepared by $A=H_1 H_2 H_5 H_6$ where $H_r$ is Hadamard gate acting on $r$-th qubit.
Remaining processes are to apply Grover iterations $Q=-A S_0 A^{-1} S_f$ with $A$ given by the state preparation operator just mentioned.
The search space examined is rather trivial, but QAA also works on arbitrary domain.
Non-trivial domain can be given as something like $N=\{ x\; |\; x \in \{0,1\}^6,\; f(x) \ne 0 \}$ for some given function $f$.
It is a matter of preparing a state encoding appropriate search space, or in other words, that is to find an operator $A$.
Once $A$ is constructed, QAA works in the same way as in Grover's algorithm.

\subsubsection{GwDP algorithm.}
GwDP algorithm is a parallelization of Grover's algorithm.
Distinguished point (DP) can be defined by a function output whose $d$ most significant bits are zeros, denoted by $d$-bit DP. We allow the notation DP to indicate a pair of DP and corresponding input or an input by itself.

For $S_q=S_c=2^s$, use $(n-2s)$-bit DP. By running $T_q=O\left(2^{n/2-s}\right)$ times of Grover iterations, DP is expected to be found on each machine. Storing $O(2^s)$ DPs sorted according to the output, a collision is found with high probability. The time-space product is always $T_qS_q=O\left(N^{1/2}\right)$.

\subsubsection{CNS algorithm.}
CNS algorithm consists of two phases, the list preparation and the collision finding.
In the list preparation phase, a list of size $2^l$ of $d$-bit DPs is drawn up
with the time complexity of $O(2^{l+d/2})$ and the classical storage of size $O(2^l)$.
In the collision finding phase QAA algorithm is used. Each iteration of QAA algorithm consists of $O(2^{d/2})$ Grover iterations and $O(2^l)$ operations for the list comparison. After $O\left(2^{(n-d-l)/2}\right)$ QAA iterations, a collision is expected to be found.
In total, CNS algorithm has $O\left(2^{l+d/2}+2^{(n-d-l)/2}(2^{d/2}+2^l)\right)$ time complexity and uses $O(2^l)$ classical memory. With the optimal parameters $l=d/2$ and $d=2n/5$, a collision is found in $T_q=O(N^{2/5})$ with $S_c=O(N^{1/5})$.

If $S_q=2^s$
, time complexity becomes $O(2^{(n-d-l-s)/2}(2^{d/2}+2^l)+2^{l+d/2-s})$ for $s \le \min(l,n-d-l)$.
When $l=d/2$ and $d=2/5 \{n+s\}$, the complexities satisfy
$(T_q)^5(S_q)^3=O(N^2)$ and $T_q(S_c)^3=O(N)$.

\subsection{AES and SHA-2 algorithms}\commentout{APPDX}
A brief review of AES and SHA-2 are given in this subsection.
Specifically, AES-128 and SHA-256 algorithms are described which will form the main body of later sections.

\subsubsection{AES-128.}
Only the encryption procedure of AES-128 which is relevant to this work will be shortly reviewed.
See \cite{nist-aes} for details.

\paragraph{Round.}
AES round consists of four elementary operations; SubBytes, ShiftRows, MixColumns, and AddRoundKey\footnote{The first and the last rounds are different, but will not be covered in detail here.}.
Each operation applies to internal state, which is represented by $4\times4$ array of bytes $S_{i,j}$, as shown in Fig.\;\ref{fig:aes_algorithm}(a).

\begin{itemize}
  \item ShiftRows does cyclic shifts of the last three rows of the internal state by different offsets.
  \item MixColumns does a linear transformation on each column of the internal state that mixes the data.
  \item AddRoundKey does an addition of the internal state and the round key by an XOR operation.
  \item SubBytes does a non-linear transformation on each byte. SubBytes works as substitution-boxes (S-box) generated by computing a multiplicative inverse, followed by a linear transformation and an addition of S-box constant.
\end{itemize}
\paragraph{Key Schedule.}
AES key schedule consists of four operations; RotWord, SubWord, Rcon, and addition by XOR operation.
The sequence of key scheduling is described in Fig.\;\ref{fig:aes_algorithm}(b).
Each operation applies to 32-bit word $w_i$, which is represented by $4\times1$ array of bytes $k_{i,j}^d$.
First four words are given by original key which become the zeroth round key.
More words$-$ 40 in AES-128$-$ are then generated by recursively processing previous words.
Every Sixteen byte $k_{i,j}^d$ constitutes $d$-th round key.
RotWord, SubWord, and Rcon only apply to every fourth word $w_i,~i \in \{3,7,11,...39\}$.

\begin{itemize}
  \item RotWord does a cyclic shift on four bytes.
  \item Rcon does an addition of the
  constant and the word by XOR operation.
  \item SubWord does an S-box operation on each byte in word.
\end{itemize}
\begin{figure}[htbp]
    \centering
    \includegraphics[width=0.90\textwidth]{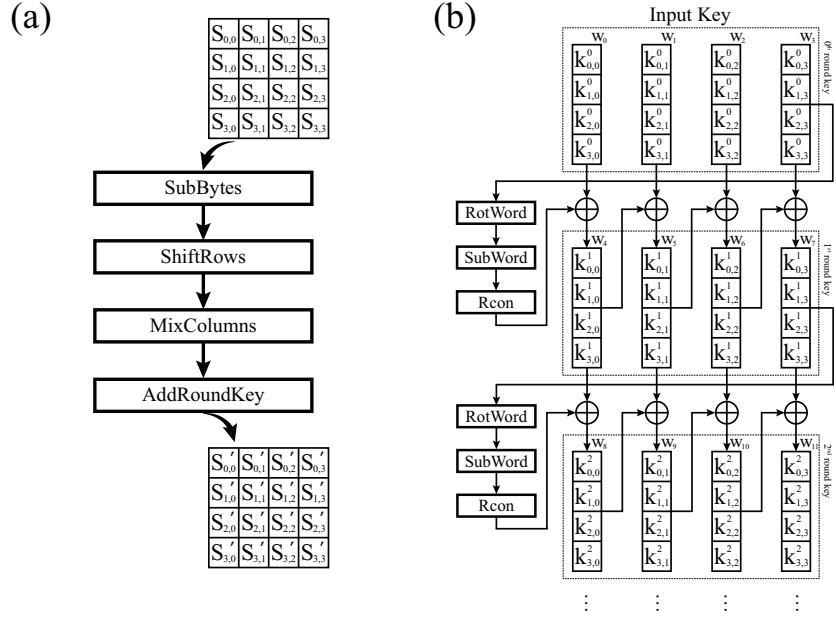}
    \caption{(a) Round operations and (b) key schedule of AES-128 algorithm. Each square box accommodates one byte. In key schedule, 128-bit key is divided into four 32-bit words.}
    \label{fig:aes_algorithm}
\end{figure}

\subsubsection{SHA-256.}
For brevity, only SHA-256 hashing algorithm for one message block which is relevant to this work will be reviewed.
Description of preprocessing including message padding, parsing, and setting initial hash value is also omitted here.
See \cite{nist-sha} for details.

\paragraph{Round.}
SHA-2 round consists of five operations; \emph{Ch, Maj,} $\Sigma_0$, $\Sigma_1$, and addition modulo $2^{32}$.
Round operations apply on eight 32-bit working variables denoted by \emph{a, b, c, d, e, f, g, h}.
See Fig.\;\ref{fig:sha2_algorithm}(a) for procedures.

\begin{itemize}
  \item $Ch(x,y,z) = (x \wedge y) \oplus (\neg x \wedge z)$,
  \item $Maj(x,y,z) = (x \wedge y) \oplus (x \wedge z) \oplus (y \wedge z)$,
  \item $\Sigma_0(x) = ROT\!R^2(x) \oplus ROT\!R^{13}(x) \oplus ROT\!R^{22}(x)$,
  \item $\Sigma_1(x) = ROT\!R^6(x) \oplus ROT\!R^{11}(x) \oplus ROT\!R^{25}(x)$,
  \item[] where $ROT\!R^n(x)$ is circular right shift of $x$ by $n$ positions.
\end{itemize}
%

\paragraph{Message Schedule.}
SHA-2 message schedule consists of three operations; $\sigma_0$, $\sigma_1$, and addition modulo $2^{32}$.
The sequence of message scheduling is described in Fig.\;\ref{fig:sha2_algorithm}(b).
Each operation applies to 32-bit word $W_i$.
First 16 words are given by original message block which become the first 16 words fed to SHA-256 rounds.
More words$-$ 48 in SHA-256$-$ are then generated by recursively processing previous words.

\begin{itemize}
  \item $\sigma_0(x) = ROT\!R^7(x) \oplus ROT\!R^{18}(x) \oplus S\!H\!R^3(x)$,
  \item $\sigma_1(x) = ROT\!R^{17}(x) \oplus ROT\!R^{19}(x) \oplus S\!H\!R^{10}(x)$,
  \item[] where $S\!H\!R^n(x)$ is right shift of $x$ by $n$ positions.
\end{itemize}
%
%
\begin{figure}[htbp]
    \centering
    \includegraphics[width=0.9\textwidth]{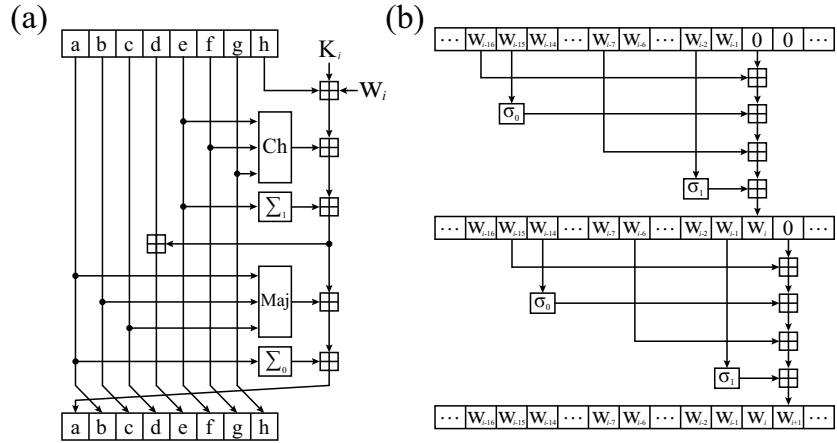}
    \caption{(a) Round operations and (b) message schedule of SHA-256 algorithm. Each square box accommodates 32-bit word. The symbol $\boxplus$ is addition modulo $2^{32}$.
    Note that new word can be overwritten on existing word if the word has already been fed to round.}
    \label{fig:sha2_algorithm}
\end{figure}

\subsection{Quantum Resource Estimates}
Quantum resource estimates of Shor's period finding algorithm have long been studied in the various literature.
See for example~\cite{adderlists,ecdlp17} and referenced materials therein.
On the other hand, quantitative quantum analysis on cryptographic schemes other than period finding is still in its early stage.
Partial list may include attacks on multivariate-quadratic problems~\cite{mq16}, hash functions~\cite{sha16,revs15}, and AES~\cite{aes16}.
We introduce two of them which are the most relevant to our work.

\subsubsection{AES Key Search.}
Grassl et al. reported the quantum costs of AES-$k$ key search for $k\in\{128,192,256\}$ in the units of logical qubit and gate~\cite{aes16}.
In estimating the time cost, the author's focus was put on a specific gate called `T' gate and its depth, although the overall gate count was also provided.
Space cost was simply estimated as the total number of qubits required to run Grover's algorithm.

There are two points we pay attention on.
First is that the authors ensured a single target key.
Since AES algorithm works like a random function, there is non-negligible probability that a plaintext end up with the same ciphertext when encrypted by two different keys.
To avoid the cases, the authors encrypt $r ~(\in \{3,4,5\})$ plaintext blocks simultaneously to obtain $r$ ciphertexts so that only the true key results in given ciphertexts.
The procedure removes the ambiguity in the number of iterations.
Note, however, that the removal of the ambiguity comes in exchange of at least tripling the space cost.
The other point is that reversible circuit implementation of internal functions of AES was always aimed at reducing the number of qubits.
One may see proposed circuit design as space-optimized.

\subsubsection{SHA-2 and SHA-3 Pre-Image Search.}
Amy et al. reported the quantum costs of SHA-2 and SHA-3 pre-image search in the units of logical and physical qubit and gate~\cite{sha16}.
The method considers an error-correction scheme called surface code.
Time cost was set considering the scheme.
Estimating the costs of T gates in terms of physical resources was one of the main results.
One point we would like to address in the work is that random-like behavior of SHA function was not considered.
It is assumed in the paper that the unique pre-image of a given hash exits.


\section{Trade-offs in Query Complexity}\label{sec:expected_iter}
The definitions of search problems involving random function are given.
For each problem, the expected iteration number
of the corresponding quantum search algorithm is calculated.
Finally, the trade-offs equations between the number of iterations and the number of quantum machines are given.

\subsection{Types of Search Problems}
We assume that $f\from X\to Y$ is random function which means $f$ is selected from the set of all functions from $X$ to $Y$ uniformly at random. Useful statistics of random function can be found in~\cite{random-mapping}. The probabilities related to the number of pre-images are quoted below. When an element $x$ is selected from a set $X$ uniformly at random, it is denoted by $x \rs X$.

When $|Y|=N$ and $|X|=aN\in\mathbb{N}$ for some $a\in\mathbb{Q}$, an element $y\in Y$ is called a \emph{$j$-node} if it has $j$ pre-images, i.e., $\left|\{x\in X:f(x)=y\}\right|=j$. For $y \rs Y$, the probability of $y$ to be a $j$-node, denoted by $q_{(aN)}(j)(j\geq0)$, is approximately
\eqnsp
\begin{align}\label{eq:jnode_prob}
    q_{(aN)}(j)=\frac{1}{e^a}\cdot\frac{a^j}{j!} \enspace.
\end{align}
For $x \rs X$, the probability of $f(x)$ to be a $j$-node, denoted by $r_{(aN)}(j)(j\geq1)$, is approximately
\eqnsp
\begin{align}\label{eq:inode-prob-domain}
    r_{(aN)}(j) = j \cdot q_{(aN)}(j) \enspace.
\end{align}

The target function in cryptanalytic search problems is usually modeled as a pseudo-random function (PRF) or a cryptographic hash function (CHF). The precise interpretation of this notions can be found in Sect. 3.5 and Sect. 5.5 of \cite{KL07}. It can be assumed that PRF and CHF have similar statistic behaviors to random function.

The formal definitions of search problems relevant to symmetric cryptanalysis can be described with random function.
The way of generating the given information in each problem is carefully distinguished. The first is key search problem which comes from the secret key search problem using a pair of plaintext and ciphertext of an encryption algorithm.

\begin{definition}[Key Search (KS)]\label{def:ks}
For a fixed $f\from X\to Y$, $y=f(x_0)$ generated from an element $x_0 \rs X$ is given. \emph{Key Search} is to find \emph{the} target $x_0$.
\end{definition}

The existence of the target $x_0$ in $X$ is always ensured. However pre-images of $y$ other than $x_0$ can be found, which is called a \emph{false alarm}. The false alarms have to be resolved by additional information since no clue (that helps to recognize the real target) is given within the problem.
Handling of false alarms is assumed not to consume quantum resources.

%

Definitions coming from the pre-image and the collision problems of CHF are given as follows.

\begin{definition}[Pre-image Search (PS)]\label{def:ps}
For a fixed $f\from\aset\to Y$, $y$ is chosen at random, $y \rs Y$, or equivalently, $y=f(x_0)$ for $x_0 \rs \aset$. \emph{Pre-image Search} is to find \emph{any} $x$ satisfying $f(x)=y$ for given $y$.
\end{definition}

There is no false alarm in pre-image search. However, the existence of a pre-image in a fixed subset of $\aset$ cannot be ensured.

\begin{definition}[Collision Finding (CF)]\label{def:cf}
For a fixed $f\from\aset\to Y$, \emph{Collision Finding} is to find \emph{any} input pair $(x_1,x_2)$ satisfying $f(x_1)=f(x_2)$.
\end{definition}

$\aset$ is imported since a domain of reasonable size including the original pre-image of a practically given hash value cannot be specified.


\subsection{Trade-offs in Grover's Algorithm for Key Search}
In this subsection, the expected iteration number and the parallelization trade-offs of Grover's algorithm are given. We assume that $|X|=|Y|=N$.

In key search problem, $y$ becomes $t$-node with probability $r(t)$ of (\ref{eq:inode-prob-domain}). The probability that one of the pre-images of $y$ is found by the measurement after $i$-times Grover iterations becomes $p_{t}(i)$ of (\ref{eq:suc_prob}). Since only one target among $t$ pre-images is the true key, the probability that the answer is correct is $1/t$. $P_{\rnd}^{\ks}(i)$ denotes the success probability after $i$-times Grover iterations of key search problem. To emphasize that $f$ is assumed to be a random function, the subscript `rand' is specified.
$P_{\rnd}^{\ks}(i)$ is the summation over possible $t$'s,
\eqnsp
\begin{align}\label{eq:suc_prob_rand_KS}
  P_{\rnd}^{\ks}(i)=\sum_{t\geq1}r(t)\cdot p_{t}(i) \cdot \frac{1}{t} \enspace.
\end{align}

Proposition about the optimal expected iterations follows.

\begin{proposition}\label{prop:iter_rand_KS}
The optimal expected number, $I_{\rnd}^{\ks}$ of Grover iterations for key search problem of random function 
becomes
\eqnsp
\begin{align*}
  I_{\rnd}^{\ks} = 0.951\ldots\cdot\sqrt{N} \enspace.
\end{align*}
\end{proposition}
\begin{proof}
This proof is similar to the one in Sect.~4 of~\cite{BBHT98}.

If the measurement is made after $i$-times Grover iterations, the expected number of iterations can be expressed as a function of $i$, denoted by $I_{\rnd}^{\ks}(i)$, which reads 
\eqnsp
\begin{align*}
  I_{\rnd}^{\ks}(i)=\frac{i}{P_{\rnd}^{\ks}(i)} \enspace.
\end{align*}
The optimal value, $I_{\rnd}^{\ks}$ is the first positive local minimum value of $I_{\rnd}^{\ks}(i)$. The first positive root of derivative of $I_{\rnd}^{\ks}(i)$, denoted by $i_{\rnd}^{\ks}$, can be calculated by a numerical analysis such as Newton's method. The result is $i_{\rnd}^{\ks} = 0.434\ldots\cdot\sqrt{N}$ and $I_{\rnd}^{\ks}=I_{\rnd}^{\ks}(i_{\rnd}^{\ks})$.$\qed$
\end{proof}

Comparing $I_{\rnd}^{\ks}$ with $I_{1}$, the expected iteration increases by 37.8\ldots\%.

The parallel trade-offs curve of key search problem is calculated in the rest of this subsection.
If inner parallelization method is taken for $S_q\gg1$, the number of pre-images of $y$ in each divided space becomes only 0 or 1 for overwhelming probability, even though $f$ is random function. Therefore the success probability after $i$-times iterations, denoted by $P_{\rnd}^{\ks:\ip}(i)$, reads
\eqnsp
\begin{align}\label{eq:prob_rand_ip}
  P_{\rnd}^{\ks:\ip}(i)=P_{\rnd,N}^{\ks:\ip}(i)=p_{1,(N/S_q)}(i) \enspace,
\end{align}
from (\ref{eq:suc_prob}). The optimal expected iteration number is similar to  (\ref{eq:iter_t_val}) as
\eqnsp
\begin{align}
  \label{eq:iter_rand_ip_val_KS}
  I_{\rnd}^{\ks:\ip} = 0.690\ldots\cdot\sqrt{(N/S_q)} \enspace.
\end{align}

In outer parallelization method, the success probability after $i$-times iterations becomes
\eqnsp
\begin{align*}
P_{\rnd}^{\ks:\op}(i)=1-\left(1-P_{\rnd}^{\ks}(i)\right)^{S_q}\enspace,
\end{align*}
and then the optimal expected iteration number for $S_q\gg1$ is given by
\eqnsp
\begin{align}
  \label{eq:iter_rand_op_val_KS}
  I_{\rnd}^{\ks:\op} = 0.784\ldots\cdot\sqrt{(N/S_q)} \enspace.
\end{align}

As a result, inner parallelization is $11.9\ldots$\% more efficient than outer method in key search problem.
We denote the number of machines used in key search problem $S_q^{\ks}$.
The optimal expected number of iterations in key search problem, denoted by $T_q^{\ks}$, can be considered as $I_{\rnd}^{\ks:\ip}$.

\begin{proposition}[KS trade-offs curve]\label{prop:to-KS}
For $S_q^{\ks}\gg1$, the parallelization trade-offs of Grover's algorithm for key search of random function is given by
\eqnsp
\begin{align*}
  (T_q^{\ks})^2S_q^{\ks}=0.476\ldots\cdot N \enspace.
\end{align*}
\end{proposition}

In the followings, the optimal expected number of iterations and trade-offs curves are defined and analyzed in the same way as in this subsection, but briefly.


\subsection{Trade-offs in Grover's Algorithm for Pre-image Search}\label{subsec:PS}
Let $X$ be the domain of the function $f$, and assume $|X|=|Y|=N$.
In pre-image search problem, there exit $t$ pre-images of $y$ with probability $q(t)$ in (\ref{eq:jnode_prob}). The success probability of measuring one of the targets after $i$-times iterations is a summation of $q(t)\cdot p_t(i)$ over $t$ as 
\eqnsp
\begin{align*}
  P_{\rnd}^{\pis}(i)=\sum_{t\geq0}q(t)\cdot p_t(i) \enspace.
\end{align*}
Since $p_0(i)=0$ and $q(t)=r(t)/t$ for $t\geq1$, it can be written as $P_{\rnd}^{\pis}(i)=P_{\rnd}^{\ks}(i)$. The important difference between the key search and the pre-image search is the existence of failure probability. If the domain of size $N$ is used, the probability there is no pre-image of $y$ in $X$ is $q(0)=1/e\approx 0.368\ldots$.

Two resolutions can be sought for the failure. The first is to change the domain $X$ in every execution of Grover's algorithm. In this case, the result on the optimal iteration number of pre-image search becomes the same as Proposition~\ref{prop:iter_rand_KS}. 
The second is to expand the domain size, $|X|=aN\in\mathbb{N}$ for some $a>1$. The success probability then reads $P_{\rnd,(aN)}^{\pis}(i) = \sum_{t\geq1}q_{(aN)}(t)\cdot p_{t,(aN)}(i)$. 

\begin{proposition}\label{prop:iter_rand_PS_ld}
If $|X|\gg N$, the optimal expected number of iterations, denoted by $I_{\rnd,(\gg N)}^{\pis}$, for pre-image search problem is written as
\eqnsp
\begin{align*}
  I_{\rnd,(\gg N)}^{\pis} = 0.690\ldots\cdot\sqrt{N} \enspace.
\end{align*}
\end{proposition}
When $N=2^{256}$, the proposition can be assumed to hold for $a\geq2^{10}$. Subscript `$\gg1$' specifies the assumption. The fact that $I_{\rnd,(\gg N)}^{\pis}\approx I_{1,N}$, i.e., better performance up to some converged value for larger domain size, is remarked. If $a$ grows to $8$, the inversely proportional failure probability decreases below $0.0004\ldots\approx1/e^8$.

In the case of inner parallelization for $|X|=|Y|$, the pre-images of $y$ are distributed to different divided spaces with overwhelming probability when $S_q\gg1$. The success probability reads 
\eqnsp
\begin{align*}
  P_{\rnd,N}^{\pis:\ip}(i)=\sum_{t\geq1}q(t)\cdot\left\{1-\left(1-p_{1,(N/S_q)}(i)\right)^t\right\} \enspace.
\end{align*}
and the optimal expected iteration number is written as
\eqnsp
\begin{align}
  \label{eq:iter_rand_ip_val_PS}
  I_{\rnd,N}^{\pis:\ip} = 0.981\ldots\cdot\sqrt{(N/S_q)} \enspace. 
\end{align}

Since $\!P_{\rnd}^{\pis}\!(i)\!\!=\!\!P_{\rnd}^{\ks}\!(i)$, the behavior of outer parallelization of pre-image search is the same as in key search. The optimal expected iteration number is
\eqnsp
\begin{align}\label{eq:iter_rand_op_val_PS}
  I_{\rnd}^{\pis:\op}=0.784\ldots\cdot\sqrt{(N/S_q)} \; \left(=I_{\rnd}^{\ks:\op} \right) \enspace. 
\end{align}

%
%

When $|X|=aN\in\mathbb{N}$ for some $a>1$, if $S_q>a^2$, it can be assumed that all pre-images of $y$ are separately distributed to the divided space in inner parallelization. For both of inner and outer parallelization, the optimal expected iteration converges to the value of (\ref{eq:iter_rand_op_val_PS}) when $a\gg1$ and $S_q>a^2$.

There are subtleties in comparing inner and outer parallelization which are inappropriate to be pointed out here.
We conclude that it is always favored to enlarge the domain size, and then for large $S_q$, two parallelization methods show asymptotically the same performance.
Denoting the optimal time and space complexities for pre-image search problem by $T_q^{\pis}$ and $S_q^{\pis}$, the trade-offs curve is given as follows.

\begin{proposition}[PS trade-offs curve]\label{prop:to-PS}
For $S_q^{\pis}\gg1$, the parallelization trade-offs of Grover's algorithm for pre-image search of random function is given by
\eqnsp
\begin{align*}
(T_q^{\pis})^2S_q^{\pis} = 0.614\ldots\cdot N \enspace.
\end{align*}
\end{proposition}

Note that while the inner parallelization is a better option in key search, both methods have similar behaviors in pre-image search.

\subsection{Trade-offs in Quantum Collision Finding Algorithms}
A collision could be found by using Grover's algorithm in the way of \emph{second pre-image} search. This has the same result as Sect. \ref{subsec:PS} if the input of the given pair of `first pre-image' is not included in the domain.
Apart from Grover's algorithm, the optimal expected iterations and trade-offs curves for parallelizations of two collision finding algorithms, GwDP and CNS, are given in this subsection.

In collision finding algorithms, searching for a pre-image of large size set is required.
Let $f\from\aset\to Y$, $X\subset\aset$ be a set of size $aN\in\mathbb{N}$. For $y \rs Y$, the expected number of pre-images of $y$ becomes $a=\sum_{j\geq1}j\cdot q_{(aN)}(j)$. If the size of a set $A\subset Y$ is large enough, it can be assumed that the number of pre-images of $A$, $|\{x\in X:f(x)\in A\}|=a\cdot|A|$.

\subsubsection{GwDP algorithm.}
Let $S_q=2^s$ and $X\subset\aset$ be a set of size $N$. In each quantum machine, a parameter $(n-2s+2)$ is used for the number of bits to be fixed in DPs.
The parameter $(n-2s+2)$ is chosen as an optimal one only among integers in order to allow the easier implementation by quantum gates.

After $i$-times Grover iterations, the success probability of measuring a DP becomes $p_{(2^{2s-2})}(i)$ from (\ref{eq:suc_prob}). The expected number of DPs found is $2^s\cdot p_{(2^{2s-2})}(i)$ by measurements after $i$-times iterations on each machine. As a result of \emph{birthday problem} (BP), known to be proposed by R. Mises in 1939, if there are $k$ samples independently selected out of $2^{2s-2}\,$ DPs, the probability of at least one coincidence, denoted by $p_{(2^{2s-2})}^{\bs}(k)$, is approximated as 
\eqnsp
\begin{align*}
  p_{(2^{2s-2})}^{\bs}(k)=1-\exp\left(\frac{-k^2}{2\cdot 2^{2s-2}}\right) \enspace.
\end{align*}
Details of approximation can be found in Sect. A.4 of~\cite{KL07}. The probability of finding at least one collision, denoted by $P_{\rnd}^{\cfg}(i)$, is then 
\eqnsp
\begin{align*}
  P_{\rnd}^{\cfg}(i)=p_{(2^{2s-2})}^{\bs}\left(2^s\cdot p_{(2^{2s-2})}(i)\right) \enspace.
\end{align*}
The optimal expected iteration reads
\eqnsp
\begin{align}\label{eq:iter_rand_GwDP}
  I_{\rnd}^{\cfg}=1.532\ldots\cdot\sqrt{N}/2^s \enspace.
\end{align}
Denoting the optimal time and space complexities by $T_q^{\cfg}$ and $S_q^{\cfg}$ for collision finding by GwDP algorithm, the trade-offs curve is given as follows.

\begin{proposition}[GwDP trade-offs curve]\label{prop:to-CF_GwDP}
 For $S_q^{\cfg}=2^s\gg1$, the trade-offs curve of {\rm GwDP} algorithm for random function is given by
\eqnsp
  \begin{align*}
    T_q^{\cfg}S_q^{\cfg}=1.532\ldots\cdot N^\frac{1}{2} \enspace.
  \end{align*}
\end{proposition}
Note that the algorithm also requires $S_c^{\cfg}=O(2^s)$ classical storage.

\subsubsection{CNS algorithm.}
In the list preparation phase, a list $L$ of size $2^l$, a subset of $d$-bit DPs, is to be made. Set $X_1\subset\aset$ of size $N=2^n$ and the function
\eqnsp
\begin{align*}
   f_{DP}(x) = \left\{
            \begin{array}{ll}
              1\enspace, & \hbox{~~if $f(x)$ is DP}\enspace, \\
              0\enspace, & \hbox{~~otherwise} \enspace.
            \end{array}
          \right.
\end{align*}
Let $f_{DP}|_{X_1}$ be the restriction of $f_{DP}$ on $X_1$.
Grover iteration in this phase is defined by $Q_{DP}=- A_{DP} S_0 A_{DP}^{-1} S_{f_{DP}|_{X_1}}$ where the oracle operator $S_{f_{DP}|_{X_1}}$ is a quantum implementation of the function $f_{DP}|_{X_1}$ and $A_{DP}$ is a usual state preparation operator $H^{\otimes n}$.

Since there are about $2^{n-d}\left(=|X_1|/2^d\right)$ DPs in $X_1$, the expected number of Grover iterations to find a DP is the same as $I_{2^{n-d}}=0.690\ldots\cdot2^{d/2}$ of (\ref{eq:iter_t_val}). The expected number of Grover iterations to build $L$ is $0.690\ldots\cdot2^{l+d/2}$. A classical storage of size $O(2^l)$ is required in addition.

In the collision finding phase, let $X_2\subset\aset$ be a set of size $N$ such that $X_1\cap X_2=\varnothing$.
Let the state $|\psi\rangle$ be an equal-phase and equal-weight superposition of states encoding all the DPs in $X_2$. State preparation operator $A_L$ such that $|\psi\rangle = A_L |0 \rangle$ is explicitly 
\eqnsp
\begin{align*}
 A_L = (- A_{DP} S_0 A_{DP}^{-1} S_{f_{DP}|_{X_2}})^{(\pi/4)\cdot2^{d/2}} A_{DP}\enspace,
\end{align*}
which is essentially Grover iterations similar to $Q_{DP}$ with repetition number $I_{2^{n-d}}^{\mxp}$ of (\ref{eq:iter_mp}).
The function $f_L\from X_2\to \{0,1\}$ is defined as
\eqnsp
\begin{align*}
   f_{L}(x) = \left\{
            \begin{array}{ll}
              1 \enspace, & \hbox{~~if $f(x)\in L$} \enspace, \\
              0 \enspace, & \hbox{~~otherwise} \enspace.
            \end{array}
          \right.
\end{align*}
To realize the oracle operator $S_{f_L}$$-$ a quantum implementation of $f_L$$-$ without a need for quantum memory, the authors have suggested a computational method taking $O(2^l)$ elementary operations per quantum $f_L$ query.

QAA iteration $Q_L$ of the collision finding phase consists of two steps. The first is acting of the oracle operator $S_{f_L}$. Let $t_L$ be the ratio of the time cost of $S_{f_L}$ per list element of $L$ to that of Grover iteration. The second step is acting of the diffusion operator $-A_LS_0A_L^{-1}$.

The success probability of QAA algorithm is known to have the same behaviors of Grover's algorithm~\cite{QAA}.
Since there are about $2^{n-d} \left(=2^n\cdot(1/2^d)\right)$ DPs encoded in the state with equal probabilities and about $2^l \left(=2^{n-d}\cdot(2^l/2^{n-d})\right)$ pre-images of $L$ in $\kett{\psi}$, by applying $Q_L$ operator $I_{(2^l),(2^{n-d})}=0.690\ldots\cdot2^{(n-d-l)/2}$ times on $\kett{\psi}$, the algorithm is expected to find a collision. 
The time cost of the collision finding phase reads $0.690\ldots\cdot2^{(n-d-l)/2}\cdot\left\{2\cdot(\pi/4)\cdot2^{d/2}+t_L\cdot2^l\right\}$. Note that the time cost of $S_0$ in collision finding phase and the initial $A_L$ are negligible.
The time cost of CNS algorithm in terms of Grover iterations denoted by $I_{\rnd}^{\cfc}(d,l)$ reads
\eqnsp
\begin{align}\label{eq:iter_rand_CNS}
  I_{\rnd}^{\cfc}(d,l) = 0.690\ldots\cdot2^{l+\frac{d}{2}} + 0.690\ldots\cdot2^{\frac{n-d-l}{2}} \left( \frac{\pi}{2}\cdot2^{\frac{d}{2}} + t_L\cdot2^l \right) \enspace.
\end{align}
The optimal value $I_{\rnd}^{\cfc}$ is given as follows.

\begin{proposition}\label{prop:iter_CF_CNS}
The optimal expected number of Grover iterations in CNS algorithm for collision finding of random function reads
\eqnsp
\begin{align*}
  I_{\rnd}^{\cfc} = 3.150\ldots\cdot t_L^{\frac{1}{5}} \cdot 2^{\frac{2}{5}n} \enspace,
\end{align*}
when $l = d/2 + \log_2 \left( \pi / (2 t_L) \right)$, and $d = 2/5 \{ n + \log_2 \left( (2 t_L)^3 / \pi  \right) \}$.
\end{proposition}

Using $S_q=2^s$ quantum machines, natural parallelization of the list preparation phase is finding $2^{l-s}$ elements on each machine.
Outer parallelization of QAA algorithm in the collision finding phase has the same expected iterations as (\ref{eq:iter_rand_op_val_PS}). The expected number of Grover iterations, denoted by $I_{\rnd}^{\cfc:\op}(d,l)$, where $s<\min(l,n-d-l)$, is written as
\eqnsp
\begin{align*}
  I_{\rnd}^{\cfc:\op}(d,l) = 0.690\ldots\cdot2^{l+\frac{d}{2}-s} + 0.784\ldots\cdot2^{\frac{n-d-l-s}{2}} \left( \frac{\pi}{2}\cdot2^{\frac{d}{2}} + t_L2^l \right) \enspace.
\end{align*}
When $l = d/2 + \log_2 \left( \pi / (2 t_L) \right)$, and $d =2/5 \{ n + s + \log_2 \left( 1.291\ldots \cdot (2 t_L)^3 / \pi  \right) \}$,
the optimal expected number of iterations reads
\eqnsp
\begin{align}\label{eq:iter_rand_CNS_para}
  I_{\rnd}^{\cfc:\op} = 3.488\ldots\cdot t_L^{\frac{1}{5}}2^{\frac{2}{5}n}2^{-\frac{3}{5}s} \enspace.
\end{align}
We denote the optimal time and space complexities by $T_q^{\cfc}$ and $S_q^{\cfc}$ for collision finding by CNS algorithm.
$T_q^{\cfc}$ can be considered as $I_{\rnd}^{\cfc:\op}$.
The trade-offs curve of CNS algorithm is then given as follows.

\begin{proposition}[CNS trade-offs curve]\label{prop:to-CF_CNS}
For $S_q^{\cfc}\gg1$, the parallelization trade-offs curve of CNS algorithm for random function is given by
\eqnsp
  \begin{align*}
    (T_q^{\cfc})^5(S_q^{\cfc})^3=(3.488\ldots)^5\cdot t_L\cdot N^2 \enspace .
  \end{align*}
\end{proposition}
The algorithms also requires the classical resource $S_c^{\cfc}=O\left(N^{1/5} (S_q^{\rm CNS})^{1/5}\right)$.
If the constant $t_L$ is determined, the time-space complexity of CNS algorithm could be derived from this trade-offs curve.

\section{Depth-Qubit Cost Metric}\label{sec:quantum_cost}
Universal quantum computers are capable of carrying out elementary logic operations such as Pauli X, Hadamard, CNOT, T, and so on.
See~\cite{QCQI} for details on quantum gates.
Implementation of any cryptographic operation in this paper is restricted such that it can only be realized by using these gates.
One may think of the restriction as a quantum version of software implementation in classical computing.
Quantum security of symmetric cryptosystems can then be estimated in units of elementary logic gates.

It is generally known that each elementary gate has different physical implementation time.
Considering various aspects of quantum computing, we suggest to simplify a measure of computation time and to ignore all the other factors or gates that complicates the analysis of quantum algorithms.

Two primary resources in quantum computing, circuit depth and qubit, can be exchanged to meet a certain attack design criteria.
Time-space complexity investigated in the previous section can be used to give an attribute `efficiency' to each and every design. 
To further quantify \textit{depth-qubit complexity} and to be able to rank the efficiency, we briefly cover the time-space trade-offs of quantum resources in this section.

\subsection{Cost Measure}\label{subsec:unit}
Difficulties often arise when it comes to setting quantum complexity measures that are physically interpretable.
There exists a number of factors making it complicate, for example totally different architecture each experimental group is pursuing.
A qubit or a certain gate may costs differently in each architecture.
It is therefore hardly possible to accurately assess operational time of each type of gate in general, and to estimates overall run time.
Despite the notable difficulty in quantifying the basic unit cost of quantum computation, a number of groups have attempted to estimate the algorithm costs in various applications~\cite{sha16,aes16,revs15}.
The cost metric varies depending on author's viewpoint.
For example, one considering the fault-tolerant computation would estimate the cost involving specific hardware implementations or error-correction schemes. 
On the other hand, one that is not to impose constraints on hardware or error-correction scheme 
would estimate the cost in logical qubits and gates.
The latter approach is adopted in this work.
Readers should keep in mind that this approach ignores the overheads introduced by fault-tolerance\footnote{Fault-tolerant cost could be in general huge, but we expect that logical cost to fault-tolerant cost conversion would be more or less uniform.}. 

High-level circuit description of Grover iteration involves not only elementary gates but also larger gates such as C$^k$NOT.
It is very unlikely that such gates can be directly operated in any realistic universal quantum computers.
Decomposition of those gates into smaller ones is thus required in practical estimates.

Determining the unit time cost is a subtle matter.
We would like to address that the simplest, yet justified time cost measure involves Toffoli gate.

\begin{definition}\label{def:cost-time}
A unit of quantum computational time cost is the time required to operate a non-parallelizable Toffoli gate.
\end{definition}

In other words, Toffoli-depth will be the time cost of the algorithm.
Three justifications can be given for the distinctness of Toffoli gates.
First, Toffoli (and single) gates are universal~\cite{deutsch89,shi02,toffoli80}.
Second, circuits consisting only of Clifford gates are not advantageous over classical computing, implying that the use of non-Clifford gates such as Toffoli is essential for quantum benefit~\cite{gottesman98}.
Third, logical Toffoli gates are the main source of time bottleneck~\cite{sha16,surface12,prx12} due to the magic state distillation process for T gates~\cite{kitaev05} which comes only from a decomposition of Toffoli gates in this work.
See for example in \cite{sha16}, the ratio of execution time in all Clifford gates to all T gates is about $0.0001\ldots$ in SHA-256
when fault-tolerance is considered.
To sum up we adopt universal Toffoli gate as the only non-Clifford gate, which is responsible for quantum speedup as well as main time bottleneck of circuits presented in this paper.
Moreover it is plausible to assume that multiple Toffoli gates can be applied to qubits simultaneously as long as their input/output qubits are independent, justifying Definition~\ref{def:cost-time}.

Space cost is estimated as a total number of logical qubits required to perform the quantum search algorithm.

\begin{definition}\label{def:cost-space}
Quantum computational space cost is the number of logical qubits required to run the entire circuit.
\end{definition}

Decomposition of a high-level circuit component into smaller ones often entails a need for additional qubits, which sometimes turn into garbage bits or get cleaned after certain operations.
Overall space cost mainly comes from these qubits.
To avoid confusion caused by terminology, we clarify five kinds of qubits.

\begin{enumerate}
  \item \textbf{Data qubits} are qubits of which the space 
      is searched by the quantum search algorithm. For example in AES-128, the size of the key space is $2^{128}$ which requires 128 data qubits.
  \item \textbf{Work qubits} are initialized qubits those assist certain operation. Whether it stays in an initialized value or gets written depends on the operation.
  \item \textbf{Garbage qubits} are previously initialized work qubits, which then get written unwanted information after a certain operation.
  \item \textbf{Output qubits} are previously initialized work qubits, which then get written the output information of a certain operation.
  \item \textbf{Oracle qubit} is a single qubit used for phase kick-back (sign change) in oracle and diffusion operators.
\end{enumerate}
There is one more type of qubit not falling into above categories; a borrowed qubit~\cite{haner17}.
The concept of the borrowed qubit is not considered in this work.
Garbage and output qubits must be re-initialized before the diffusion of Grover iteration to be disentangled from data qubits.

\subsection{Toffoli and T Gates} \commentout{APPDX}
Commonly acknowledged universal quantum gate set consists of Clifford gates and T gate.
As stated in the previous subsection, operational time costs of different gates may vary depending on architecture.
However, it is less disputable that physical implementation of T gate (or preparation of the magic state) is important, difficult, and generally more expensive than Clifford gates.
There are communities dedicated to better implementation of T~\cite{bravyi12,kitaev05,jones13,knill13} and reducing the number of the gates applied~\cite{roetteler13EC,roetteler13,svore12,mosca14,ecdlp17}, as it is time bottleneck in fault-tolerant quantum computing.

Toffoli gate is a non-Clifford gate that is composed of a few T and Clifford gates.
Taking Toffoli gate over T gate as a basic unit of time resource has its merits and demerits.
We cautiously compare the relation between Toffoli and T to the one between high- and low-level languages.
Example of implementation of a two-bit addition in terms of Toffoli and T gates are given in Fig.\;\ref{fig:Toffoli_vs_T}.
\begin{figure}[htbp]
    \centering
    \includegraphics[width=0.95\textwidth]{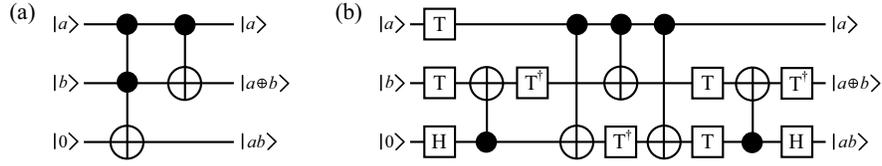}
    \caption{Addition of two bits $a$ and $b$ in terms of (a) Toffoli and (b) T gates (Fig. 7(d) in \cite{roetteler13}). The third qubit (output qubit) is written a carry. The third and the second qubits save the binary representation of $a+b$ as $ab \cdot 2^1 + (a\oplus b) \cdot 2^0$.}
    \label{fig:Toffoli_vs_T}
\end{figure}

Being reminded that Toffoli and CNOT operate as 
\eqnsp
\begin{align*}
{\rm TOF} |a \rangle |b \rangle |0 \rangle = |a \rangle |b \rangle |0\oplus ab \rangle \enspace,~
{\rm CNOT} |a \rangle |b \rangle = |a \rangle |a \oplus b \rangle \enspace,
\end{align*}
respectively, it is immediately noticeable from Fig.\;\ref{fig:Toffoli_vs_T}(a) that the circuit works as a two-bit addition operator.
The same operation realized by depth-optimized Clifford+T set~\cite{roetteler13} is described in Fig.\;\ref{fig:Toffoli_vs_T}(b).
Assuming that a given quantum computer can only perform gates in Clifford+T set, this circuit enables more transparent expectation of runtime.

Typically in previous studies a quantum algorithm is first implemented in Toffoli-level, and then the circuit undergoes a kind of `compilation' process that looks for an elementary-level circuit~\cite{sha16,ecdlp17}.
Finding an optimal compiling method is very complicated and worth researching~\cite{revs15}.
At this stage however, it is hardly possible to find true optimal elementary-level circuit from compiling huge high-level circuit.
In this work therefore, we stay in Toffoli-level implementation conforming the purpose of providing a general framework.

\subsection{Time-Space Trade-offs}\label{subsec:trade-offs}
Readers those are familiar with quantum circuit model can safely skip over this subsection as it covers some general facts about depth-qubit trade-offs.
In quantum circuit model, it is often possible to sacrifice efficiency in qubits for better performance in time and vice versa.
Quantum version of such time-space trade-offs forms a main body of Sect(s).~\ref{sec:AES} and~\ref{sec:SHA}.
As preliminary we give an example to introduce the general concept of trade-offs in quantum circuits.

Consider a function $f$ that carries out binary multiplications of $k$ single bit values.
At the end of this subsection we will deal with general $k$, but for now, let us explicitly write down the description with $k=2$, the multiplication of two bits $a$ and $b$ as $f(a, b) = ab$.
It is noticeable that the function $f$ can be implemented by AND gate in a classical circuit.
However in a quantum circuit where only unitary operations are allowed, similar implementation is prohibited since AND operation is not unitary as the input information $a$, $b$ cannot be retrieved by knowing $ab$ only.
Simple resolution can be found by keeping the input information all the way such that
\eqnsp
\begin{align}\label{eq:mul_2}
  U_f |a \rangle |b \rangle |0 \rangle = |a \rangle |b \rangle |0 \oplus ab \rangle \enspace ,
\end{align}
where $|a \rangle$ and $|b \rangle$ are quantum states encoding $a$ and $b$, and $U_f$ is the quantum implementation of the function $f$.
Previously zeroed qubit represented by the state $|0 \rangle$ on the left-hand side holds the result after the operation.
There exists a quantum gate that exactly performs the operation by $U_f$ called a $k$-fold controlled-NOT (C$^k$NOT) with $k=2$, or better known as Toffoli gate. Figure\;\ref{fig:cknot}(a) illustrates the graphical representation of C$^2$NOT gate achieving (\ref{eq:mul_2}).
General C$^k$NOT gates read $k$ input bits carried by wires intersecting with black dots and change a target bit carried by a wire intersecting with exclusive-or symbol.
In this case, the gate works as NOT on target bit if $a=b=1$ and identity otherwise.
\begin{figure}[htbp]
    \centering
    \includegraphics[width=0.8\textwidth]{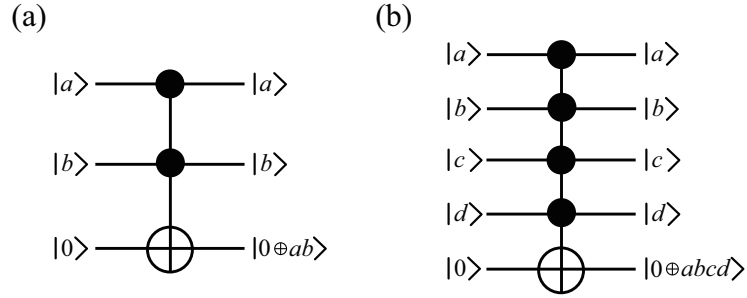}
    \caption{(a) C$^2$NOT (Toffoli) gate and (b) C$^4$NOT gate}
    \label{fig:cknot}
\end{figure}

Similarly, multiplications of four bits can be implemented by using C$^4$NOT gate as shown in Fig.\;\ref{fig:cknot}(b).
C$^4$NOT gate carries out NOT operation on target bit if $a=b=c=d=1$ and nothing otherwise.

Now assume we are to split up a C$^4$NOT gate into multiple Toffoli gates with the help of a few extra qubits.
Decomposing a large gate into smaller gates is a typical task one confront in compilation~\cite{revs15}.
There can be various ways to achieve the goal, and one of the immediate designs is the one in Fig.\;\ref{fig:cknot_split}(a).
\begin{figure}[htbp]
    \centering
    \includegraphics[width=0.95\textwidth]{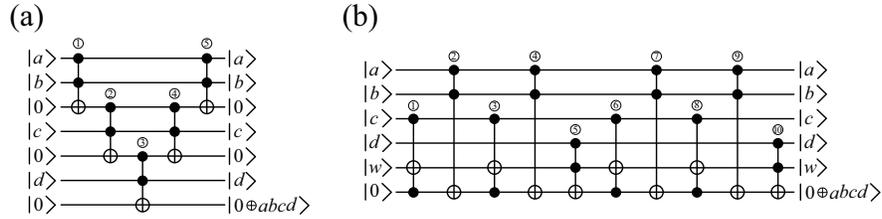}
    \caption{Decomposition of C$^4$NOT gate into (a) five Toffoli gates and (b) ten Toffoli gates. In (a), the third and the fifth zeroed qubits from the top are work qubits whereas in (b), only the fifth arbitrary-valued qubit is a work qubit.}
    \label{fig:cknot_split}
\end{figure}

Let us examine the action of each Toffoli gate on the register one-by-one,
\eqnsp
\begin{equation}\label{eq:decom_time}
\centering
\begin{tabular}{r l c l}
  & $|a \rangle |b \rangle |0 \rangle |c \rangle |0 \rangle |d \rangle |0 \rangle$
  & $\overset{\textrm{\circled{1}}}{\mapsto}$
  &$|a \rangle |b \rangle |ab \rangle |c \rangle |0 \rangle |d \rangle |0 \rangle$ \\
  $\overset{\textrm{\circled{2}}}{\mapsto}$~~
  &$|a \rangle |b \rangle |ab \rangle |c \rangle |abc \rangle |d \rangle |0 \rangle$
  &~~$\overset{\textrm{\circled{3}}}{\mapsto}$~~
  &$|a \rangle |b \rangle |ab \rangle |c \rangle |abc \rangle |d \rangle |abcd \rangle$ \\
  $\overset{\textrm{\circled{4}}}{\mapsto}$~~
  &$|a \rangle |b \rangle |ab \rangle |c \rangle |0 \rangle |d \rangle |abcd \rangle$
  &$\overset{\textrm{\circled{5}}}{\mapsto}$
  &$|a \rangle |b \rangle |0 \rangle |c \rangle |0 \rangle |d \rangle |abcd \rangle$ \enspace,
\end{tabular}
\end{equation}
where the circled number above the mapping arrow indicates the corresponding Toffoli gate in Fig.\;\ref{fig:cknot_split}(a).
The result actually comes out after $\circled{3}$, but we further perform a kind of un-computation with two extra Toffoli gates to re-initialize the work qubits.
It is up to users to decide whether the procedure should stop just after $\circled{3}$ at the cost of two garbage qubits being generated, or go all the way to the end of the circuit.
As one can notice, it is already the trade-offs.

A less straightforward decomposition can be found in Fig.\;\ref{fig:cknot_split}(b).
It makes use of twice as many Toffoli gates as Fig.\;\ref{fig:cknot_split}(a) but requires only a single arbitrary qubit\footnote{The first Toffoli gate in Fig.\;\ref{fig:cknot_split}(b) is redundant in this case, but needed if one wants to carry out $z \oplus abcd$, where $z$ is the initial value of the last qubit.}.
Each Toffoli gate transforms the state as follows.
\eqnsp
\begin{equation}\label{eq:decom_space}
\centering
\begin{tabular}{r l c l}
  & $|a \rangle |b \rangle |c \rangle |d \rangle |w \rangle |0 \rangle$
  &$\overset{\textrm{\circled{1},\circled{2}}}{\mapsto}$
  & $|a \rangle |b \rangle |c \rangle |d \rangle |w \rangle |ab \rangle$ \\
  $\overset{\textrm{\circled{3}}}{\mapsto}$~~
  & $|a \rangle |b \rangle |c \rangle |d \rangle |w\oplus abc \rangle |ab \rangle$
  &~~$\overset{\textrm{\circled{4}}}{\mapsto}$~~
  & $|a \rangle |b \rangle |c \rangle |d \rangle |w\oplus abc \rangle |0 \rangle$ \\
  $\overset{\textrm{\circled{5}}}{\mapsto}$~~
  & \multicolumn{3}{l}{$|a \rangle |b \rangle |c \rangle |d \rangle |w\oplus abc \rangle |abcd \oplus dw \rangle$} \\
  $\overset{\textrm{\circled{6}}}{\mapsto}$~~
  & \multicolumn{3}{l}{$|a \rangle |b \rangle |c \rangle |d \rangle |w\oplus abc \oplus abcd \oplus cdw\rangle |abcd \oplus dw \rangle$} \\
  $\overset{\textrm{\circled{7}}}{\mapsto}$~~
  & \multicolumn{3}{l}{$|a \rangle |b \rangle |c \rangle |d \rangle |w\oplus abc \oplus abcd \oplus cdw\rangle |abcd \oplus dw \oplus ab \rangle$} \\
  $\overset{\textrm{\circled{8}}}{\mapsto}$~~
  & \multicolumn{3}{l}{$|a \rangle |b \rangle |c \rangle |d \rangle |w \rangle |abcd \oplus dw \oplus ab \rangle$} \\
  $\overset{\textrm{\circled{9}}}{\mapsto}$~~
  &$|a \rangle |b \rangle |c \rangle |d \rangle |w \rangle |abcd \oplus dw \rangle$ &
  $~~\overset{\textrm{\circled{\tiny{1\!0}}}}{\mapsto}$~~
  &$|a \rangle |b \rangle |c \rangle |d \rangle |w \rangle |abcd \rangle$ \enspace.
\end{tabular}
\end{equation}

Both designs work as desired.
In fact for general $k$, time-efficient design as in Fig.\;\ref{fig:cknot_split}(a) requires $k-2$ zeroed work qubits within depth $2k-3$, whereas space-efficient design as in Fig.\;\ref{fig:cknot_split}(b) uses only one arbitrary qubit within depth $8k-24$ (for $k \ge 5$)~\cite{barenco95}.
We denote time- and space-efficient designs lower-depth and less-qubit C$^k$NOT, respectively.

Bit multiplication is one of examples qubit and depth are mutually exchangeable.
In Sect(s).~\ref{sec:AES} and~\ref{sec:SHA} we will compare multiple circuits that do the same job with a different number of qubits, and examine the consequence of each design when parallelized.


\section{Complexity of AES-128 Key Search}\label{sec:AES}
This section presumes that readers are familiar with standard AES-128 encryption algorithm~\cite{nist-aes}.
We assume that a quantum adversary is given a plaintext-ciphertext pair and asked to find the key used for the encryption.
Since AES-128 works as a PRF, it is possible that multiple keys lead to the same ciphertext,
\eqnsp
\begin{align*}
   \textrm{AES}(k_0,p) = \textrm{AES}(k_1,p) = \cdots \enspace ,
\end{align*}
where $k_i \in \{0,1\}^{128}$ are different keys and $p$ is a given plaintext.
The term pre-image will be used to denote each key $k_i$ that generates given ciphertext upon the encryption of given plaintext.

\begin{figure}[htbp]
    \centering
    \includegraphics[width=0.8\textwidth]{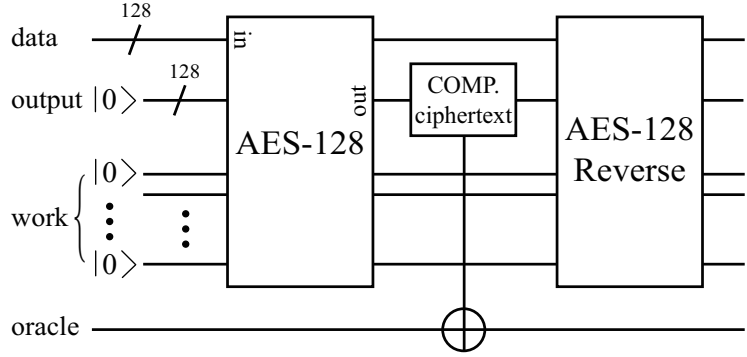}
    \caption{Oracle circuit for key search attack on AES-128}
    \label{fig:oracle-AES}
\end{figure}

The idea of applying Grover's algorithm to exhaustive attack on AES-128 is as follows.
Linearly superposed $2^{128}$ input keys encoded in 128 data qubits are fed as an input to an AES box shown in Fig.\;\ref{fig:oracle-AES}.
AES box contains a reversible circuit implementation of AES-128 encryption algorithm.
The AES box encrypts the given plaintext, outputting superposed ciphertexts encoded in output qubits.
Superposed ciphertexts are then compared with given ciphertext via C$^{128}$NOT gate to mark the target.
After marking is done, every qubit except the oracle qubit is passed on to a reverse AES box to disentangle the data qubits from other qubits.

\subsection{Circuit Implementation Cost}
AES-128 encryption internally performs SubBytes, MixColumns, ShiftRows, AddRoundKey, SubWord, RotWord and Rcon.
Quantum circuits for these operations are mostly adopted from \cite{aes16} with improvements and fixes.

MixColumns, ShiftRows and RotWord are linear operations acting on 32 bits that do not require any work qubit nor Toffoli gate.
Among them, last two are simple bit permutations which require no quantum gates (by re-wiring) or at most SWAP gates.\commentout{APPDX}
MixColumns needs to be treated more carefully as it is not a bit permutation.
Treating each four-byte column of the internal state as a length-four vector, MixColumns is expressed as a matrix multiplication,
\eqnsp
\begin{align}\label{eq:MixColumns}
  \left( \begin{array}{c}
      s_{0,j}' \\
      s_{1,j}' \\
      s_{2,j}' \\
      s_{3,j}' \\
    \end{array} \right)
  =
  \left( \begin{array}{cccc}
      02 & 03 & 01 & 01 \\
      01 & 02 & 03 & 01 \\
      01 & 01 & 02 & 03 \\
      03 & 01 & 01 & 02 \\
    \end{array} \right)
  \left( \begin{array}{c}
      s_{0,j} \\
      s_{1,j} \\
      s_{2,j} \\
      s_{3,j} \\
    \end{array} \right),
  \quad \textrm{for } 0 \le j \le 3 \enspace,
\end{align}
where 01, 02, 03 are sub-matrices when each byte $s_{i,j}$ is treated as a length-eight vector, written as
\eqnsp
\begin{align*}
  01 = \left( \begin{array}{cccccccc}
           1 & 0 & 0 & 0 & 0 & 0 & 0 & 0 \\
           0 & 1 & 0 & 0 & 0 & 0 & 0 & 0 \\
           0 & 0 & 1 & 0 & 0 & 0 & 0 & 0 \\
           0 & 0 & 0 & 1 & 0 & 0 & 0 & 0 \\
           0 & 0 & 0 & 0 & 1 & 0 & 0 & 0 \\
           0 & 0 & 0 & 0 & 0 & 1 & 0 & 0 \\
           0 & 0 & 0 & 0 & 0 & 0 & 1 & 0 \\
           0 & 0 & 0 & 0 & 0 & 0 & 0 & 1 \\
         \end{array} \right), \;
  02 = \left( \begin{array}{cccccccc}
           0 & 1 & 0 & 0 & 0 & 0 & 0 & 0 \\
           0 & 0 & 1 & 0 & 0 & 0 & 0 & 0 \\
           0 & 0 & 0 & 1 & 0 & 0 & 0 & 0 \\
           1 & 0 & 0 & 0 & 1 & 0 & 0 & 0 \\
           1 & 0 & 0 & 0 & 0 & 1 & 0 & 0 \\
           0 & 0 & 0 & 0 & 0 & 0 & 1 & 0 \\
           1 & 0 & 0 & 0 & 0 & 0 & 0 & 1 \\
           1 & 0 & 0 & 0 & 0 & 0 & 0 & 0 \\
         \end{array} \right), \;
  03 = \left( \begin{array}{cccccccc}
           1 & 1 & 0 & 0 & 0 & 0 & 0 & 0 \\
           0 & 1 & 1 & 0 & 0 & 0 & 0 & 0 \\
           0 & 0 & 1 & 1 & 0 & 0 & 0 & 0 \\
           1 & 0 & 0 & 1 & 1 & 0 & 0 & 0 \\
           1 & 0 & 0 & 0 & 1 & 1 & 0 & 0 \\
           0 & 0 & 0 & 0 & 0 & 1 & 1 & 0 \\
           1 & 0 & 0 & 0 & 0 & 0 & 1 & 1 \\
           1 & 0 & 0 & 0 & 0 & 0 & 0 & 1 \\
         \end{array} \right) \enspace.
\end{align*}
Since an explicit form of transformation matrix is given in (\ref{eq:MixColumns}), the quantum circuit implementation of the matrix can be found by methods given in~\cite{roetteler01,markov08}.

AddRoundKey and Rcon are XOR-ings of fixed-size strings which can also be efficiently realized by CNOT or X gates only.

\begin{figure}[htbp]
    \centering
    \includegraphics[width=0.7\textwidth]{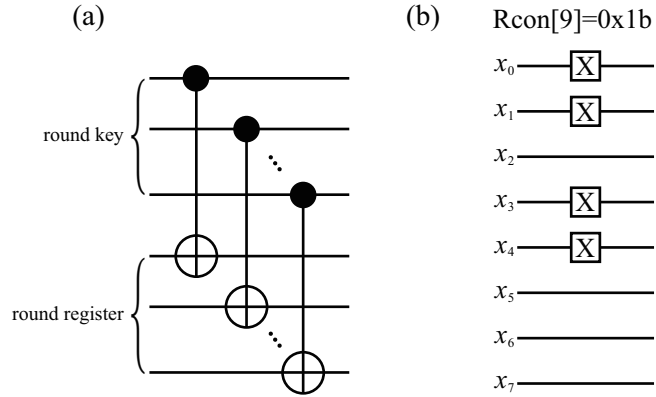}
    \caption{Quantum circuit implementation of (a) AddRoundKey and (b) Rcon. In (b), Rcon[9] circuit is given as an example.}
    \label{fig:ARK_Rcon}
\end{figure}

SubBytes and SubWord are the only operations which require quantum resources.
Since SubBytes and SubWord consist of 16 and 4 S-boxes, the S-box is the only operation to be carefully discussed.

Classically, S-box can be implemented as a look-up table.
However, a quantum counterpart of such table should involve the notion of the quantum memory aforementioned in Sect. \ref{subsec:variants}.
Therefore in this work, S-box is realized by explicitly calculating multiplicative inverse followed by GF-linear mapping and addition of the S-box constant as described in Sect. 3.2.1 of~\cite{aes16}.

S-box is realized by calculating multiplicative inverse followed by GF-linear mapping and addition of S-box constant.\commentout{APPDX}
By treating a byte as an element in GF$(2^8)= \textrm{GF}(2)[x]/(x^8+x^4+x^3+x+1)$, GF-linear mapping and addition of S-box constant are summarized as the equation
\eqnsp
\begin{align}\label{eq:sbox-affine}
   \left(
     \begin{array}{c}
       x_0' \\
       x_1' \\
       x_2' \\
       x_3' \\
       x_4' \\
       x_5' \\
       x_6' \\
       x_7' \\
     \end{array}
   \right)
   =
   \left(
     \begin{array}{cccccccc}
       1 & 0 & 0 & 0 & 1 & 1 & 1 & 1 \\
       1 & 1 & 0 & 0 & 0 & 1 & 1 & 1 \\
       1 & 1 & 1 & 0 & 0 & 0 & 1 & 1 \\
       1 & 1 & 1 & 1 & 0 & 0 & 0 & 1 \\
       1 & 1 & 1 & 1 & 1 & 0 & 0 & 0 \\
       0 & 1 & 1 & 1 & 1 & 1 & 0 & 0 \\
       0 & 0 & 1 & 1 & 1 & 1 & 1 & 0 \\
       0 & 0 & 0 & 1 & 1 & 1 & 1 & 1 \\
     \end{array}
   \right)
   \left(
     \begin{array}{c}
       x_0 \\
       x_1 \\
       x_2 \\
       x_3 \\
       x_4 \\
       x_5 \\
       x_6 \\
       x_7 \\
     \end{array}
   \right)
+
   \left(
     \begin{array}{c}
       1 \\
       1 \\
       0 \\
       0 \\
       0 \\
       1 \\
       1 \\
       0 \\
     \end{array}
   \right) \enspace,
\end{align}
where addition is XOR operation and $x_i$ are coefficients of polynomial of order $x^7$.
No work qubit nor Toffoli gate is required in this step.
While XOR operation is simply done by applying X gates to relevant qubits, implementing a transformation matrix in (\ref{eq:sbox-affine}) is not trivial.
See~\cite{roetteler01,markov08} for general methods of realizing linear transformations.

Resource estimate of quantum AES-128 encryption has been narrowed down to estimate the cost of finding multiplicative inverse of the element $\alpha$ in GF($2^8$).
In~\cite{aes16}, multiplicative inverse of $\alpha$ is calculated by using two arithmetic circuits; Maslov et al.'s modular multiplier~\cite{maslov09} and in-place squaring~\cite{aes16}.
Slight modification of previous method is found in this work with seven multipliers being used, verified by the quantum circuit simulation by matrix product state~\cite{mps}, with the seven multipliers being used as following sequences,
\eqnsp
\begin{align}\label{eq:inverse}
\begin{array}{>{\centering}p{0.95cm} l >{\centering}p{0.95cm} l >{\centering}p{0.95cm} l >{\centering}p{0.95cm} l >{\centering}p{0.95cm} l >{\centering}p{0.95cm} l}   
  &|\alpha \rangle      &                & 
   |\alpha \rangle      &                & 
   |\alpha \rangle      &                & 
   |\alpha \rangle      &                & 
   |\alpha \rangle      &                & 
   |\alpha \rangle                       \\
  &|0 \rangle           & \mapstr{CNOTs} & 
   |0 \rangle           & \mapstr{Sq}    & 
   |0 \rangle           & \mapstr{Mul}   & 
   |0 \rangle           & \mapdstr{Sq$^{-1}$}{CNOTs}{0.4} & 
   |0 \rangle           & \mapstr{Sq$\times2$}            & 
   |0 \rangle                            \\
  &|0 \rangle           &                & 
   |\alpha \rangle      &                & 
   |\alpha^2 \rangle    &                & 
   |\alpha^2 \rangle    &                & 
   |0        \rangle    &                & 
   |0        \rangle                     \\
  &|0 \rangle           &                & 
   |0 \rangle           &                & 
   |0 \rangle           &                & 
   |\alpha^3 \rangle    &                & 
   |\alpha^3 \rangle    &                & 
   |\alpha^{12} \rangle                  \\
  &|0 \rangle           &                & 
   |0 \rangle           &                & 
   |0 \rangle           &                & 
   |0 \rangle           &                & 
   |0 \rangle           &                & 
   |0 \rangle                            \\
  &&&&&&&&&&& \\
  &|\alpha \rangle      &                & 
   |\alpha \rangle      &                & 
   |\alpha \rangle      &                & 
   |\alpha^{64} \rangle &                & 
   |\alpha^{64} \rangle &                & 
   |\alpha \rangle                       \\
  \mapdstr{CNOTs}{Sq$\times2$}{0.4}
  &|0 \rangle           & \mapdstr{Mul}{Sq$^{-1}\!\!\times\!\!2$}{0.4}  &
   |0 \rangle           & \mapdstr{CNOTs}{Sq$^{-1}\!\!\times\!\!2$}{0.4}&
   |0 \rangle           & \mapdstr{Sq$\times6$}{Mul}{0.4}               &
   |0 \rangle           & \mapdstr{Mul}{Sq}{0.4}                        &
   |\alpha^{254} \rangle& \mapdstr{Sq$^{-1}\!\!\times\!\!6$}{Mul$^{-1}$}{0.4} &
   |\alpha^{254} \rangle                 \\
  &|0 \rangle           &                & 
   |\alpha^{60} \rangle &                & 
   |\alpha^{60} \rangle &                & 
   |\alpha^{60} \rangle &                & 
   |\alpha^{60} \rangle &                & 
   |\alpha^{60} \rangle                  \\
  &|\alpha^{12} \rangle &                & 
   |\alpha^{12} \rangle &                & 
   |\alpha^3 \rangle    &                & 
   |\alpha^3 \rangle    &                & 
   |\alpha^3 \rangle    &                & 
   |\alpha^3 \rangle                     \\
  &|\alpha^{48} \rangle &                & 
   |\alpha^{12} \rangle &                & 
   |0 \rangle           &                & 
   |\alpha^{63} \rangle &                & 
   |\alpha^{63} \rangle &                & 
   |0 \rangle                            \\
  &&&&&&&&&&& \\
  &|\alpha \rangle      &                & 
   |\alpha \rangle      &                & 
   |\alpha \rangle      &                & 
   |\alpha \rangle      &                & 
   |\alpha \rangle      &                & 
   |\alpha \rangle                       \\
  \mapdstr{Sq$\times2$}{CNOTs}{0.4}
  &|\alpha^{254\!}\rangle\!& \mapdstr{Sq$\times2$}{Mul$^{-1}$}{0.4}         &
   |\alpha^{254\!}\rangle\!& \mapdstr{Sq$^{-1}\!\!\times\!\!2$}{CNOTs}{0.4} &
   |\alpha^{254\!}\rangle\!& \mapdstr{CNOTs}{Sq$^{-1}\!\!\times\!\!2$}{0.4} &
   |\alpha^{254\!}\rangle\!& \mapdstr{Sq}{Mul$ ^{-1}$}{0.4}                 &
   |\alpha^{254\!}\rangle\!& \mapdstr{Sq$^{-1}$}{CNOTs}{0.4}                &
   |\alpha^{254\!}\rangle                 \\
  &|\alpha^{60} \rangle &                & 
   |0 \rangle           &                & 
   |0 \rangle           &                & 
   |\alpha \rangle      &                & 
   |\alpha^2 \rangle    &                & 
   |0 \rangle                            \\
  &|\alpha^{12} \rangle &                & 
   |\alpha^{12} \rangle &                & 
   |\alpha^{12} \rangle &                & 
   |\alpha^3 \rangle    &                & 
   |0 \rangle           &                & 
   |0 \rangle                            \\
  &|\alpha^{12} \rangle &                & 
   |\alpha^{48} \rangle &                & 
   |0 \rangle           &                & 
   |0 \rangle           &                & 
   |0 \rangle           &                & 
   |0 \rangle \enspace,                            
\end{array}
\end{align}
where each state ket represents eight-bit register, Sq and Mul denote modular squaring and multiplication operations, and CNOTs implies eight CNOT gates copying the string.
Seven multipliers including reverse operations have been used as can be seen from (\ref{eq:inverse}).

As squaring in GF($2^8$) is linear, it does not involve the use of Toffoli nor work qubits.
Therefore it is only required to estimate the cost of multipliers.
Table~\ref{tab:aes-elem-cost} summarizes the elementary operation costs in AES-128.
Two distinct multipliers are considered in this work; Maslov et al.'s design~\cite{maslov09} and Kepley and Steinwandt's design~\cite{kepley15}.

\renewcommand{\arraystretch}{1.2}
\begin{table}[htbp]
  \caption{Costs of elementary operations in AES-128. A quarter of work qubits needed in S-box turn into garbage qubits.}
  \label{tab:aes-elem-cost}
  \centering
  \begin{tabular}{>{\centering}p{2cm}   ||
                  >{\centering}p{1.7cm} | >{\centering}p{1.7cm} |
                  >{\centering}p{1.7cm} | >{\centering}p{1.7cm} }             \hline
                      & \multicolumn{2}{c|}{Less-qubit} & \multicolumn{2}{c}{Lower-depth} \tabularnewline \cline{2-5}
                      &   Multiplier   &    S-box       &   Multiplier  &    S-box        \tabularnewline \hline
     Toffoli-depth    &$      18      $&$      126     $&$      8      $&$      56      $ \tabularnewline
     Work qubits      &$      8       $&$      32      $&$      27     $&$      108      $ \tabularnewline
    \hline
  \end{tabular}
\end{table}
\renewcommand{\arraystretch}{1.0}
First four multiplications in S-box are aimed at computing the multiplicative inverse.
Remaining three (reverse) multiplications are then used to clean garbage qubits produced by previous multiplications.
At the end of S-box, a quarter of total work qubits needed in S-box turn into garbage qubits.

\subsection{Design Candidates}
Four main trade-offs points are considered.

First point, that has an impact on the overall design, is to determine whether key schedule and AES rounds are carried out in parallel.
As S-box is used both in key schedule and AES round, schedule-round parallel implementation would require more work qubits.
This option is denoted by \emph{serial/parallel schedule-round}.

Second, AES round functions can be reversed in the middle of encryption process to save work qubits.
The idea of reverse AES round was suggested in Sect. 3.2.3 in~\cite{aes16}.
Since each run of round function produces garbage qubits, forward running of 10 rounds accumulates $\ge 1280$ garbage qubits.
Putting reverse rounds in between forward rounds reduce a large amount of work qubits at the cost of longer Toffoli-depth.
This option is denoted by \emph{reverse round} when applied.

Thirdly, a choice of multiplier could make an important trade-offs point.
Less-qubit and lower-depth multipliers are two options.
For simplicity, we do not consider adaptive use of both multipliers although it is possible to improve the efficiency by using appropriate multiplier in different part of circuit.
This option is denoted by \emph{less-qubit/lower-depth multiplier}.

Fourth, to present the extremely depth-optimized circuit design, the cleaning process in S-box could be skipped leaving every work qubit used in S-box garbage.
This option is denoted by \emph{S-box un-cleaning} when applied.

In total, there exist 16 (=$2^4$) different circuit designs.
We only take six of them into account as others seem to be flawed compared with the six. Six designs are denoted as follows.

\begin{itemize}
  \item \acct{1}: Serial schedule-round, reverse-round, less-qubit multiplier
  \item \acct{2}: Serial schedule-round, reverse-round, lower-depth multiplier
  \item \acct{3}: Parallel schedule-round, less-qubit multiplier
  \item \acct{4}: Parallel schedule-round, lower-depth multiplier
  \item \acct{5}: Parallel schedule-round, less-qubit multiplier, S-box un-cleaning
  \item \acct{6}: Parallel schedule-round, lower-depth multiplier, S-box un-cleaning
\end{itemize}

\subsection{Comparison}\label{subsec:aes-comp}
Toffoli-depth and total number of qubits are carefully estimated for each design.
Costs of quantum AES-128 encryption circuit and entire Grover's algorithm on a single quantum processor is summarized in Table\;\ref{tab:aes-single-comp}. 
Estimates for single Grover iteration is omitted from the table as it can easily be calculated from costs of AES-128 encryption circuit;
\eqnsp
\begin{align*}
 \textrm{cost(Grover iteration)} = 2\cdot \textrm{cost(AES-128)} + 2\cdot\textrm{cost(C$^{128}$NOT)}\enspace,
\end{align*}
where cost(C) is Toffoli-depth of a circuit C.
Note that full Toffoli-depth of the entire Grover's algorithm is estimated considering $I_\textrm{rand}^\textrm{KS}$ in Proposition~\ref{prop:iter_rand_KS}.
\renewcommand{\arraystretch}{1.2}
\begin{table}[htbp]
  \caption{Costs of AES-128 encryption circuit and entire attack circuit on a single quantum processor. MAXDEPTH is not considered. }
  \label{tab:aes-single-comp}
  \centering
  \begin{tabular}{>{\centering}p{2cm} || >{\centering}p{2cm} |
                  >{\centering}p{2cm} | >{\centering}p{2.5cm} | >{\centering}p{2cm} }             \hline
             &\multicolumn{2}{c}{AES-128}\vline& \multicolumn{2}{c}{Grover}      \tabularnewline \cline{2-5}
             & Toffoli-depth &     Qubits      & Toffoli-depth      &    Qubits  \tabularnewline \hline
   \acct{1} &$    11088    $&$      984      $&$ 1.360\ldots\times2^{78}$&     985    \tabularnewline
   \acct{2} &$     4928    $&$     3017      $&$ 1.290\ldots\times2^{77}$&    3018    \tabularnewline
   \acct{3} &$     1260    $&$     2208      $&$ 1.405\ldots\times2^{75}$&    2209    \tabularnewline
   \acct{4} &$      560    $&$     7148      $&$ 1.510\ldots\times2^{74}$&    7149    \tabularnewline
   \acct{5} &$      720    $&$     6654      $&$ 1.808\ldots\times2^{74}$&    6655    \tabularnewline
   \acct{6} &$      320    $&$    21854      $&$ 1.064\ldots\times2^{74}$&   21855    \tabularnewline
    \hline
  \end{tabular}
\end{table}
\renewcommand{\arraystretch}{1.0}
\renewcommand{\arraystretch}{1.2}
\begin{table}[htbp]
  \caption{Comparison of time-space complexity of different AES-128 circuit designs without the oracle assumption. The smallest $c_\#^\ks$ is found by \acct{4} with $c_4^\ks = 1.048\ldots \times 2^{33}$. Other values are divided by $c_4^\ks$ for easier comparison.}
  \label{tab:aes-para-comp}
  \centering
  \begin{tabular}{>{\centering}p{1.5cm} || >{\centering}p{1.5cm} | >{\centering}p{1.5cm} |
                  >{\centering}p{1.5cm}  | >{\centering}p{1.5cm} | >{\centering}p{1.5cm} |
                  >{\centering}p{1.5cm}    }             \hline
            & \acct{1} & \acct{2} & \acct{6} & \acct{5} & \acct{3} & \acct{4}  \tabularnewline \hline
  $c_\#^\ks/c_4^\ks
           $&$ 28.606\ldots  $&$ 19.705\ldots  $&$ 1.519\ldots   $&$ 1.333\ldots   $&$ 1.070\ldots   $&$    1    $ \tabularnewline \hline
  \end{tabular}
\end{table}
\renewcommand{\arraystretch}{1.0}

Proposition~\ref{prop:to-KS} basically sets up the criterion for a comparison of circuit designs.
Here we replace $T_q^{\ks}$ and $S_q^{\ks}$ by $\mathcal{T}_q^\ks$ and $\mathcal{S}_q^\ks$, respectively, denoting Toffoli-depth and total number of qubits in key search problem, i.e.,
\eqnsp
\begin{align}\label{eq:TS-wo-bo-KS}
  \left( \mathcal{T}_q^\ks \right)^2 \mathcal{S}_q^\ks = c_\#^\ks N \enspace,
\end{align}
where $c_\#^\ks$ varies depending on circuit designs.
Now a parameter $c_\#^\ks$ is the only `yardstick' that tells us which design is better.
When parallelized for large $S_q$, the expected iteration number converges to the one given in (\ref{eq:iter_rand_ip_val_KS}).
Taking the converged value, the $c_\#^\ks$ for each circuit design is summarized in Table\;\ref{tab:aes-para-comp}.
Assuming the MAXDEPTH is capped at some fixed value smaller than $\sqrt{N}$, the table indicates that for example \acct{1} requires about 28.6\ldots times as many qubits as \acct{4}.

\subsection{Comparison to Ensured Single Target}
It is possible to guarantee an existence of a single target by using multiple plaintext-ciphertext pairs.
To ensure a single target, the oracle now performs $r$ AES encryptions simultaneously.
In~\cite{aes16}, $r=3$ is chosen for AES-128.
Each AES box encrypts different plaintext with the same superposed input keys.
As a result, for example for $r=3$ in AES-128, the probability that two pre-images exist is the same as for $k_1$ to exist such that
\eqnsp
\begin{align*}
   \textrm{AES}(k_0,p_1)\|\textrm{AES}(k_0,p_2)\|\textrm{AES}(k_0,p_3) \!=\! \textrm{AES}(k_1,p_1)\|\textrm{AES}(k_1,p_2)\|\textrm{AES}(k_1,p_3) \enspace ,
\end{align*}
where $\|$ is concatenation, $k_0$ is the true key and $p_i$ are distinct plaintexts.
The cost of guaranteeing a single target is more or less multiplying the total number of qubits by $r$.
%
%
\renewcommand{\arraystretch}{1.2}
\begin{table}[htbp]
  \caption{Comparison of attack design with and without a single target. Toffoli-depth of encryption circuit is the same in both, because the same AES design is implemented.
  It is noticeable that the full Toffoli-depth of \emph{Unique Key} is not far different from that of \acct{4}, although the number of qubits is nearly doubled. }
  \label{tab:aes-unique-comp}
  \centering
  \begin{tabular}{>{\centering}p{1.9cm} || >{\centering}p{2cm}  | >{\centering}p{1.5cm} |
                  >{\centering}p{2.5cm} | >{\centering}p{1.5cm} } \hline
             &\multicolumn{2}{c|}{AES-128}& \multicolumn{2}{c}{Grover}    \tabularnewline \cline{2-5}
             & Toffoli-depth &     Qubits &     Toffoli-depth    & Qubits  \tabularnewline \hline
   Unique Key&$      560    $&$    14296 $&$ 1.269\ldots \times2^{74} $&  14297  \tabularnewline
   \acct{4}  &$      560    $&$     7148 $&$ 1.510\ldots \times2^{74} $&   7149  \tabularnewline
    \hline
  \end{tabular}
\end{table}
\renewcommand{\arraystretch}{1.0}

It is now natural to ask if the oracle operator with a single target is more cost-efficient than the random function oracle with less qubits.
Assuming $r=2$ guarantees a single target, we compare a design dubbed \emph{Unique Key} with \acct{4}.
\emph{Unique Key}'s encryption circuit design is chosen to be the same as \acct{4}, meaning that the difference in efficiency solely comes from ensuring a single target.
Results are summarized in Table\;\ref{tab:aes-unique-comp}.
Full Toffoli-depth of \emph{Unique Key} is estimated considering $I_{1}$ in (\ref{eq:iter_t_val}).
With a guaranteed single target, Toffoli-depth is expected to be shortened compared with \acct{4} at the cost of doubling qubits.
Although ensuring single target can be regarded as the optimization point when using single processor, it \textit{strictly} cannot be an option in parallel attack since the inner parallelization removes a penalty of random characteristics as in~(\ref{eq:prob_rand_ip}).



\section{Complexity of SHA-256 Pre-image Search}\label{sec:SHA}
This section presumes that readers are familiar with standard SHA-256~\cite{nist-sha}.
The idea of applying Grover's algorithm to pre-image attack on SHA-256 is as follows.

\begin{figure}[htbp]
    \centering
    \includegraphics[width=0.8\textwidth]{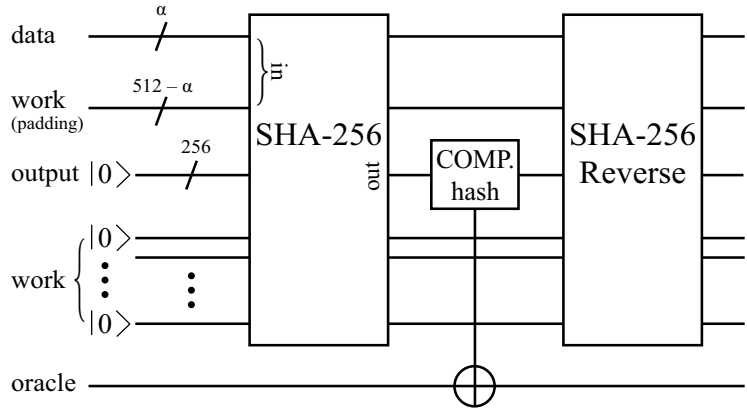}
    \caption{Oracle circuit for pre-image attack on SHA-256}
    \label{fig:oracle-SHA}
\end{figure}

A message block consisting of $\alpha$ bits of message and $512-\alpha$ bits of padding are input to the hash box as shown in Fig.\;\ref{fig:oracle-SHA}.
Hash box contains a reversible circuit implementation of SHA-256 to permit superposed input.
The input of linearly superposed $2^\alpha$ messages are then passed on to the hash box resulting in superposed corresponding hashes.
Processed hashes are then compared with the given hash via C$^{256}$NOT gate.
After the target is marked, the entire qubits except the oracle qubit are further processed through the reverse hash box as in~Fig.\;\ref{fig:oracle-SHA}.
The quantum state of the data qubits at the end of Fig.\;\ref{fig:oracle-SHA} reads
\eqnsp
\begin{align*}
  | \psi \rangle &= \frac{1}{\sqrt{2^\alpha}}
    \left(
      |00\cdots0 \rangle \!+\! |00\cdots1 \rangle \!+\! \cdots
                         \!-\! |t_i \rangle + \cdots
                         \!+\! |11\cdots1 \rangle
    \right) \otimes |\text{padding} \rangle \enspace, 
\end{align*}
where each ket state encodes a message and $t_i$'s are pre-images of the given hash value.
The number of targets probabilistically varies depending on $\alpha$ which is capped at $447(=512-64-1)$.

\subsection{Circuit Implementation Cost}
SHA-256 internally performs five elementary operations, $\sigma_{0(1)}$, $\Sigma_{0(1)}$, \emph{Ch}, \emph{Maj}, and \emph{ADDER} (modular addition)~\cite{nist-sha}.

Among internal operations carried out in SHA-256, $\Sigma_{0(1)}$ consists only of XOR-ings of bit permutations.
Results of three \emph{ROTR} operations are written on 32-bit output register, with being successively XOR-ed.
Only CNOT gates are involved in implementation with 32 work qubits.

Similarly, $\sigma_{0(1)}$ is implemented with one difference from $\Sigma_{0(1)}$, that is \emph{SHR}.
\emph{SHR} itself is not linear, but writing a result of \emph{SHR} on 32-bit output register is possible.
Therefore, $\sigma_{0(1)}$ is also efficiently realized by CNOT gates with 32 work qubits.

\emph{Ch} and \emph{Maj} are bit-wise operations that do require Toffoli gates.
We adopt Amy et al.'s design where \emph{Ch} and \emph{Maj} require one and two Toffoli gates, respectively.
See Figs. 4 and 5 in~\cite{sha16}.

Serial schedule-round implementation of SHA-256 is illustrated in Fig.~\ref{fig:circuit-SHA}.

\begin{figure}[htbp]
    \centering
    \includegraphics[width=0.99\textwidth]{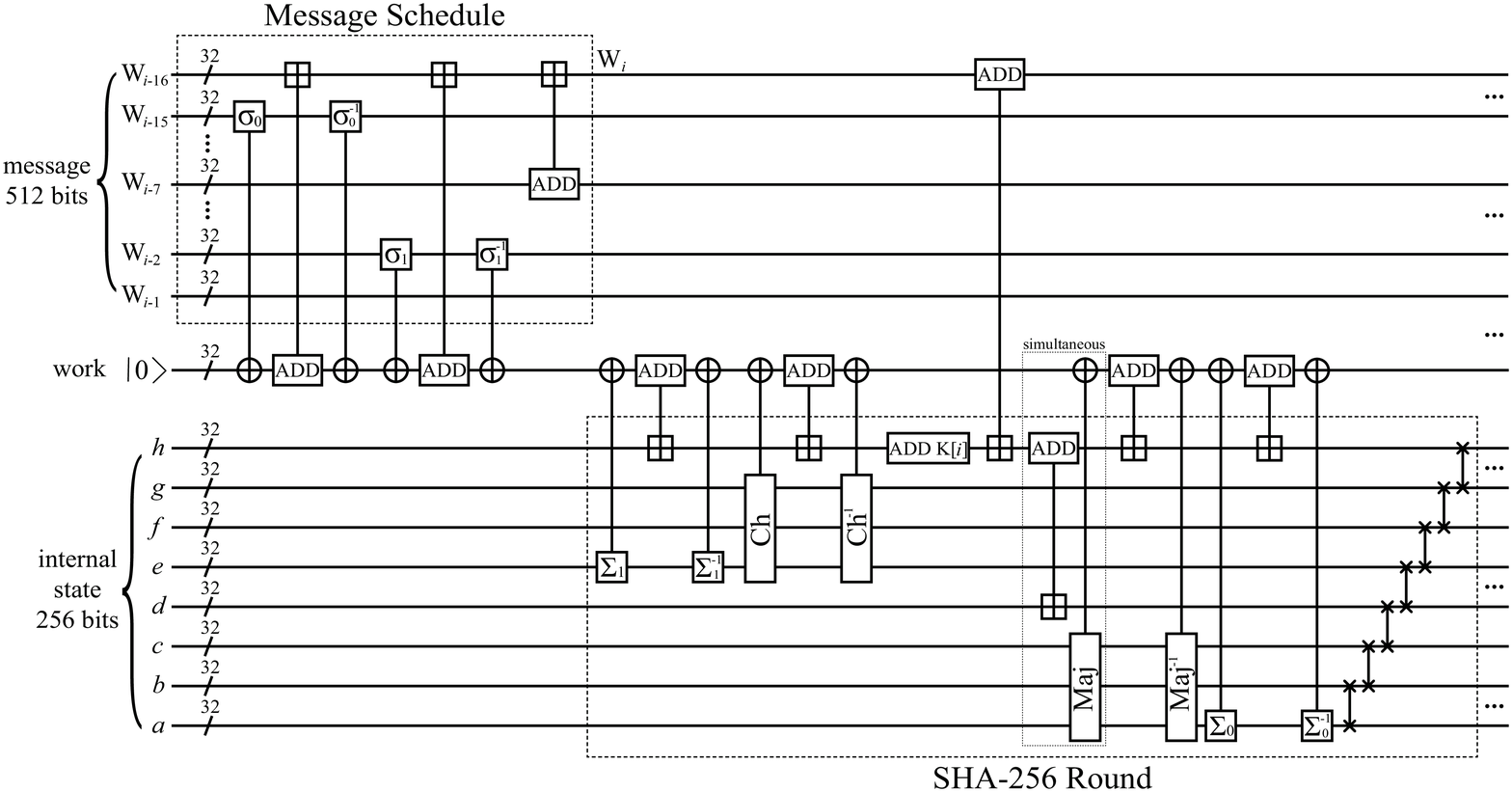}
    \caption{Reversible circuit for serial implementation of SHA-256 message schedule and round function.
    The message block consisting of 16 words are recursively updated in-place.
    Note that it is straightforward to make message schedule and round functions work in parallel by expanding the work space. Seven two-qubit gates at the end of round are SWAP gates. The symbol $\boxplus$ is addition modulo $2^{32}$.}
    \label{fig:circuit-SHA}
\end{figure}

Low-level circuit design for each function in this work is mostly adopted from~\cite{sha16} except \emph{ADDER} choice and totally re-designed message schedule.
A few options are available for \emph{ADDER} circuits one can adopt (see for example,~\cite{adderlists}).
For our purpose of comparing various circuit designs, we choose two versions of adders; a poly-depth \emph{ADDER}~\cite{adder-poly} and a log-depth \emph{ADDER}~\cite{adder-log}.
Table~\ref{tab:sha-elem-cost} summarizes resource costs of elementary operations in SHA-256.
\renewcommand{\arraystretch}{1.2}
\begin{table}[htbp]
   \caption{Costs of elementary operations in SHA-256. Work qubits in \emph{ADDER} columns get cleaned within the respective Toffoli-depth.
   Outputs of $\sigma_{0(1)}$, $\Sigma_{0(1)}$, $Ch$, and $Maj$ are written on work qubits.}
   \label{tab:sha-elem-cost}
   \centering
   \begin{tabular}{>{\centering}p{2cm} || >{\centering}p{2.2cm} |
                   >{\centering}p{2cm}    | >{\centering}p{2cm} |
                   >{\centering}p{1cm}    | >{\centering}p{1cm} }    \hline
                    & \emph{ADDER} (poly) & \emph{ADDER} (log) &$ \sigma_{0(1)}$, $\Sigma_{0(1)} $& \emph{Ch} & \emph{Maj} \tabularnewline \hline
      Toffoli-depth &$        61         $&$         22       $&$                0               $&$    1    $&$    2     $\tabularnewline
      Work qubits   &$         1         $&$         53       $&$               32               $&$    32   $&$    32    $\tabularnewline
      \hline
   \end{tabular}
\end{table}
\renewcommand{\arraystretch}{1.0}

\subsection{Design Candidates}
Three optimization points are considered.

First point, that has an impact on the overall design, is to determine whether message schedule and round functions are carried out in parallel.
Figure~\ref{fig:circuit-SHA} 
 shows a serial circuit implementation of SHA-256.
In the algorithm description, $i$-th round function is fed by $i$-th word from the schedule meaning that parallel implementation is possible if enough work qubits are given.
This option is denoted by \emph{serial/parallel schedule-round}.

Second point is to determine which \textit{ADDER} is to be used.
Use of the poly-depth \textit{ADDER} is better in saving work space whereas the log-depth \textit{ADDER} could shorten the execution time.
For simplicity, we do not consider adaptive use of both \textit{ADDER}s although it is possible to improve the efficiency by using appropriate \textit{ADDER} in different part of circuit.
This option is denoted by \emph{poly-depth/log-depth ADDER}.

Lastly, it is now optional to decide how many work qubits are to be used to implement C$^{256}$NOT gate for marking the targets (hash comparison).
As discussed in Sect.~\ref{subsec:trade-offs}, C$^k$NOT gate can be one of the trade-offs points.
However in AES-128, we do not need to consider C$^{128}$NOT as an optimization point seriously since the encryption process accompanies enough number of work qubits that can be reused in lower-depth C$^{128}$NOT gate.
Situation is different in SHA-256.
It is noticeable that hashing process of SHA-256 does not involve as many work qubits as AES-128, meaning that the lower-depth C$^{256}$NOT gate cannot be implemented unless more qubits are introduced solely for hash comparison.
Toffoli-depth and work qubits required for lower-depth (less-qubit) C$^{256}$NOT gate are 509 (2024) and 254 (1), respectively.
Note that lower-depth and less-qubit C$^{k}$NOT gates present here are only two extreme exemplary designs.
This option is denoted by \emph{less-qubit/lower-depth} C$^{256}$NOT.

In total, there exist 8 (=$2^3$) distinct circuit designs.
We only analyze six of them since others do not seem to have merits.
Six designs are denoted as follows.

\begin{itemize}
  \item \scct{1}: Serial schedule-round, poly-depth \emph{ADDER}, less-qubit C$^{256}$NOT
  \item \scct{2}: Serial schedule-round, log-depth \emph{ADDER}, less-qubit C$^{256}$NOT
  \item \scct{3}: Serial schedule-round, log-depth \emph{ADDER}, lower-depth C$^{256}$NOT
  \item \scct{4}: Parallel schedule-round, poly-depth \emph{ADDER}, less-qubit C$^{256}$NOT
  \item \scct{5}: Parallel schedule-round, log-depth \emph{ADDER}, less-qubit C$^{256}$NOT
  \item \scct{6}: Parallel schedule-round, log-depth \emph{ADDER}, lower-depth C$^{256}$NOT
\end{itemize}

\subsection{Comparison}
Toffoli-depth and total number of qubits are carefully estimated for each design.
The number of data qubits $\alpha$ has to be determined at this point.
In our numerical calculation, $\alpha=266$ seems to safely achieve the optimal expected iteration number given by Proposition \ref{prop:iter_rand_PS_ld} and to remove the failure probability.
Costs of quantum SHA-256 hashing circuit and the entire Grover's algorithm on a single quantum processor is summarized in Table\;\ref{tab:sha-single-comp}. 
Estimates for single Grover iteration is omitted from the table as it can easily be calculated from costs of SHA-256 circuit;
\eqnsp
$$
\textrm{cost(Grover iteration)} = 2\cdot \textrm{cost(SHA-256)} + \textrm{cost(C$^{256}$NOT)} + \textrm{cost(C$^{266}$NOT)}\enspace.
$$

\renewcommand{\arraystretch}{1.2}
\begin{table}[htbp]
  \caption{Costs of SHA-256 hashing circuit and entire attack circuit on a single quantum processor. MAXDEPTH is not considered.}
  \label{tab:sha-single-comp}
  \centering
  \begin{tabular}{>{\centering}p{2cm} || >{\centering}p{2cm} |
                  >{\centering}p{2cm} | >{\centering}p{2.5cm} | >{\centering}p{2cm} }             \hline
             &\multicolumn{2}{c|}{SHA-256} & \multicolumn{2}{c}{Grover}      \tabularnewline \cline{2-5}
             & Toffoli-depth &     Qubits  &      Toffoli-depth  &    Qubits  \tabularnewline \hline
   \scct{1} &$    36368    $&$      801  $&$ 1.586\ldots\times2^{143}$&     802    \tabularnewline
   \scct{2} &$    13280    $&$      853  $&$ 1.227\ldots\times2^{142}$&     854    \tabularnewline
   \scct{3} &$    13280    $&$      853  $&$ 1.163\ldots\times2^{142}$&    1023    \tabularnewline
   \scct{4} &$    27584    $&$      834  $&$ 1.216\ldots\times2^{143}$&     835    \tabularnewline
   \scct{5} &$    10112    $&$      938  $&$ 1.919\ldots\times2^{141}$&     939    \tabularnewline
   \scct{6} &$    10112    $&$      938  $&$ 1.792\ldots\times2^{141}$&    1023    \tabularnewline
    \hline
  \end{tabular}
\end{table}
\renewcommand{\arraystretch}{1.0}

\renewcommand{\arraystretch}{1.2}
\begin{table}[htbp]
  \caption{Comparison of trade-offs coefficients of different SHA-256 circuit designs without the oracle assumption. The smallest $c_\#^\pis$ is found by \scct{6} with $c_6^\pis = 1.034\ldots \times 2^{38}$. Other values are divided by $c_6^\pis$ for easier comparison.}
  \label{tab:sha-para-comp}
  \centering
  \begin{tabular}{>{\centering}p{1.5cm} || >{\centering}p{1.5cm} | >{\centering}p{1.5cm} |
                  >{\centering}p{1.5cm}  | >{\centering}p{1.5cm} | >{\centering}p{1.5cm} |
                  >{\centering}p{1.5cm}    }             \hline
            & \scct{1} & \scct{4} & \scct{3} & \scct{2} & \scct{5} & \scct{6}  \tabularnewline \hline
  $c_\#^\pis/c_6^\pis$
            &$  9.830\ldots  $&$  6.015\ldots $&$  1.685\ldots   $&$  1.565\ldots  $&$  1.053\ldots  $&$    1    $ \tabularnewline \hline
  \end{tabular}
\end{table}
\renewcommand{\arraystretch}{1.0}

Proposition~\ref{prop:to-PS} establishes the criterion for the comparison.
Similar to~(\ref{eq:TS-wo-bo-KS}), we replace $T_q^{\pis}$ and $S_q^{\pis}$ by $\mathcal{T}_q^{\pis}$ and $\mathcal{S}_q^{\pis}$, respectively, denoting Toffoli-depth and total number of qubits, i.e.,
\eqnsp
\begin{align}\label{eq:TS-wo-bo-PS}
  \left( \mathcal{T}_q^{\pis} \right)^2 \mathcal{S}_q^{\pis} = c_\#^{\pis} N \enspace,
\end{align}
where $c_\#^\pis$ varies depending on efficiency of circuits.
When parallelized for large $S_q$, the expected iteration number converges to the one given in (\ref{eq:iter_rand_op_val_PS}).
Taking the converged value, $c_\#^\pis$ for each design is summarized in Table\;\ref{tab:sha-para-comp}.
If MAXDEPTH is capped at some fixed value smaller than $\sqrt{N}$, the table indicates that for example \scct{1} requires about 9.8\ldots times as many qubits as \scct{6}.

\section{Complexity of SHA-256 Collision Finding}\label{sec:SHA-coll}
Costs of two collision finding algorithms, GwDP and CNS, are to be estimated in this section.
We adopt \scct{6} which also turn out to be the most efficient in time-space complexity in GwDP and CNS algorithms\footnote{Details on circuit comparisons in GwDP and CNS algorithms are dropped from the main text. An interesting point worth noticing is that \scct{5} has small advantageous range of $S_q (< 2^{8})$ over \scct{6}. The reason is that while \scct{5} requires zero additional qubit in hash comparison, \scct{6} needs $(256-d-2)$ qubits in comparison where $d$ is the number of fixed bits in DP. Since $d$ grows as $S_q$ increases, there occurs crossover point. It is also noticeable that \scct{6} cannot exactly fit into Proposition \ref{prop:to-CF_GwDP} for the same reason just mentioned, but deviation is small.}.

\subsection{GwDP Algorithm}
Estimating the cost of GwDP algorithm is straightforward.
Basically this algorithm constructs a set of DPs by running multiple instances of Grover's algorithm so that there occurs collision in the set.
By using (\ref{eq:iter_rand_GwDP}), costs of GwDP algorithm for selected number of machines are summarized in Table\;\ref{tab:GwDP-para}.

\renewcommand{\arraystretch}{1.2}
\begin{table}[htbp]
  \caption{Costs of GwDP algorithm for various number of machines. Note that the algorithm also requires classical memory of size $O(S_q)$.
  }
  \label{tab:GwDP-para}
  \centering
  \begin{tabular}{>{\centering}p{1cm} || >{\centering}p{2.5cm} | >{\centering}p{2.5cm}}  \hline
    $S_q$   &   Toffoli-depth    &   Qubits  \tabularnewline \hline
    $2^2$   &$1.986\ldots\times2^{141}$&$    4084         $ \tabularnewline \hline
    $2^4$   &$1.985\ldots\times2^{139}$&$   16272         $ \tabularnewline \hline
    $2^8$   &$1.984\ldots\times2^{135}$&$   258304        $ \tabularnewline \hline
    $2^{16}$&$1.981\ldots\times2^{127}$&$ 6.508\ldots\times10^7    $ \tabularnewline \hline
    $2^{32}$&$1.975\ldots\times2^{111}$&$ 4.127\ldots\times10^{12} $ \tabularnewline \hline
    $2^{64}$&$1.963\ldots\times2^{79} $&$ 1.732\ldots\times10^{22} $ \tabularnewline \hline
  \end{tabular}
\end{table}
\renewcommand{\arraystretch}{1.0}

If $T_q^{\cfg}$ and $S_q^{\cfg}$ in Proposition \ref{prop:to-CF_GwDP} are replaced by Toffoli-depth $\mathcal{T}_q^{\cfg}$ and number of qubits $\mathcal{S}_q^{\cfg}$, the trade-offs curve reads
\eqnsp
\begin{align}\label{eq:TS-wo-bo-CF}
  \mathcal{T}_q^{\cfg} \mathcal{S}_q^{\cfg} = c^{\cfg} \cdot N^{\frac{1}{2}} \enspace ,
\end{align}
where $c^{\cfg}$ is found to be $1.802\ldots\times2^{25}$ by using $S_q=2^{64}$ case.

\subsection{CNS Algorithm}\label{subsec:CNS}
Proposition~\ref{prop:iter_CF_CNS} suggests the optimal expected number of iterations in terms of $t_L$.
The only extra work need to be done here is to determine $t_L$ explicitly.
From the definition of $t_L$, it reads
\eqnsp
\begin{gather}
  t_L = \frac{ {\rm cost}(S_{f_L}) }{ 2^l\cdot {\rm cost}(G) } \enspace, \nn
\begin{align}
  &\textrm{cost}(S_{f_L}) = 2 \cdot \textrm{cost(SHA-256)} + 2^l \cdot \textrm{cost}\left( \textrm{C$^{(256-d)}$NOT} \right) \enspace, \nn
  &\textrm{cost}(G) = 2 \cdot \textrm{cost(SHA-256)} + \textrm{cost}\left( \textrm{C$^d$NOT} \right) + \textrm{cost}\left( \textrm{C$^{256}$NOT} \right) \enspace, \nonumber
%
%
\end{align}\nn
  l = \frac{d}{2} + \log_2 \left( \frac{\pi}{2 t_L} \right) \enspace, \qquad
  d= \left\lfloor \frac{512}{5} + \frac{2}{5} \log_2\left( \frac{(2 t_L)^3 }{\pi } \right) \right\rceil \enspace,      \label{eq:cns_t_L}
\end{gather}
where $G$ is Grover iteration.
Numerical approach was taken to find $t_L$, $d$ and $l$, which came out to be $0.015182\ldots$, $96$ and $54.538\ldots$, respectively.
By substituting these values for parameters in (\ref{eq:iter_rand_CNS}), the expected number of iterations becomes $I_{\rnd}^{\cfc} = 1.856\ldots \times 2^{102}$.
Note that this value is somewhat different from that of Proposition \ref{prop:iter_CF_CNS} as $d$ has been rounded off. 
Finally by multiplying $I_{\rnd}^{\cfc}$ and the time cost of $G$, we obtain the total Toffoli-depth of CNS algorithm as
\eqnsp
\begin{align}
  I_{\rnd}^{\cfc} \cdot {\rm cost}(G) = 1.184\ldots \times 2^{117} \enspace .
\end{align}
%
Quantum space cost is cheaper than \scct{6} because C$^{(256-d)}$NOT gate used for list-comparison requires less work qubits than C$^{256}$NOT in pre-image search.
It is estimated to be 939 qubits in total.
\renewcommand{\arraystretch}{1.2}
\begin{table}[htbp]
  \caption{Parameter values and costs of CNS algorithm for various number of machines. Note that the algorithm also requires $O(N^{1/5} S_q^{1/5})$ classical resources. }
  \label{tab:cns-para}
  \centering
  \begin{tabular}{>{\centering}p{1cm} || >{\centering}p{1.5cm} | >{\centering}p{1cm} |
                  >{\centering}p{1.8cm}  | >{\centering}p{2.5cm} | >{\centering}p{2.5cm} }  \hline
    $S_q$ &   $l$    & $d$   &   $t_L$    &    Toffoli-depth   &  Qubits  \tabularnewline \hline
    $2^2$ & $55.155\ldots$ & $97$  & $0.015064\ldots$ &$1.353\ldots\times2^{116}$&$       3756         $ \tabularnewline \hline
    $2^4$ & $55.558\ldots$ & $98$  & $0.014987\ldots$ &$1.203\ldots\times2^{115}$&$      15024         $ \tabularnewline \hline
    $2^8$ & $56.364\ldots$ & $99$  & $0.014834\ldots$ &$1.729\ldots\times2^{112}$&$      240384        $ \tabularnewline \hline
  $2^{16}$& $57.976\ldots$ & $102$ & $0.014527\ldots$ &$1.960\ldots\times2^{107}$&$    6.154\ldots\times10^{7}  $ \tabularnewline \hline
  $2^{32}$& $61.201\ldots$ & $109$ & $0.013914\ldots$ &$1.352\ldots\times2^{98} $&$    4.033\ldots\times10^{12} $ \tabularnewline \hline
  $2^{64}$& $67.654\ldots$ & $121$ & $0.012692\ldots$ &$1.100\ldots\times2^{79} $&$    1.732\ldots\times10^{22} $ \tabularnewline \hline
  \end{tabular}
\end{table}
\renewcommand{\arraystretch}{1.0}

When parallelized, $t_L$ slightly changes since $l$ and $d$ depend on $S_q(=2^s)$, the number of machines.
Modified $l$ and $d$ reads
\eqnsp
\begin{gather*}
%
  l = \frac{d}{2} + \log_2 \left( \frac{\pi}{2 t_L} \right) \enspace , \quad
  d= \left\lfloor \frac{512 + 2s}{5} + \frac{2}{5} \log_2\left( \frac{1.291\ldots (2 t_L)^3 }{\pi } \right) \right\rceil \enspace,
\end{gather*}
where $t_L$, cost($S_{f_L}$) and cost($G$) are the same as in (\ref{eq:cns_t_L}).
We have estimated the quantum resource costs of CNS algorithm for a few $S_q$ values as summarized in Table\;\ref{tab:cns-para}.
Note that estimated time complexities are different from ones given by (\ref{eq:iter_rand_CNS_para}) as the equation is obtained for large $S_q$, and $d$ here has been rounded off to the nearest integer.
Due to the bound $s<\min(l,n-d-l)$, $S_q=2^{66}$ is almost the maximum number of quantum machines Proposition \ref{prop:to-CF_CNS} holds.


%


\section{Security Strengths of AES and SHA-2}\label{sec:sec-str}
Based on the results of previous sections, quantum security strengths of AES and SHA-2 are drawn in this section.
Three MAXDEPTH parameters, $2^{40}$, $2^{64}$, and $2^{96}$, are adopted from~\cite{nist-quantum}.
Note that using these values of MAXDEPTH in our analysis is a conservative approach as our estimates only count Toffoli gates as time resources whereas NIST has counted all gates.
Security strength of SHA-2 is determined by collision finding problem, not by pre-image search problem.

\renewcommand{\arraystretch}{1.2}
\begin{table}[htbp]
  \caption{Trade-offs coefficients of AES-$k$ key search problem for $k\in\{128,192,256\}$ and SHA-$m$ collision finding problem for $m\in\{256,384,512\}$. Coefficients $c_k^{\ks}$ and $c_m^{\cf}$ are divided by their respective minimal values $c_{128}^{\ks} = c_4^{\ks}$ and $c_{256}^{\cf} = c^{\cfg}$. }
  \label{tab:longer-key}
  \centering
  \begin{tabular}{>{\centering}p{1.5cm} || >{\centering}p{2cm} | >{\centering}p{2cm} | >{\centering}p{2cm}}  \hline
                  &  AES-128 &  AES-192 &  AES-256    \tabularnewline \hline
    $c_k^{\rm KS}/c_{128}^{\rm KS}$
                  &$    1   $&$  1.560\ldots $&$  2.586\ldots $   \tabularnewline \hline
  \end{tabular}
  \begin{tabular}{>{\centering}p{1.5cm} || >{\centering}p{2cm} | >{\centering}p{2cm} | >{\centering}p{2cm}}  \hline
                  &  SHA-256 &  SHA-384  & SHA-512    \tabularnewline \hline
    $c_m^{\rm CF}/c_{256}^{\rm CF}$
                  &$    1   $&$  3.837\ldots $&$  3.940\ldots $   \tabularnewline \hline
  \end{tabular}
\end{table}
\renewcommand{\arraystretch}{1.0}

Resource estimates for AES-128 key search problem with circuit \acct{4} is extended to AES-192 and AES-256, and similarly that of SHA-256 collision finding problem with circuit \scct{6} is applied to SHA-384 and SHA-512.
Since the depth-qubit trade-offs curves (\ref{eq:TS-wo-bo-KS}) and (\ref{eq:TS-wo-bo-CF}) must holds for larger key and message digest sizes, we only compare their trade-offs coefficients in Tables\;\ref{tab:longer-key}.
There is a tendency that the values of coefficients grow as the key or message digest sizes get larger.
Increasing coefficient values reflect various complexity factors added; more rounds, longer schedule, larger word size, and so on.
Especially in hash, size of the message block in SHA-384 is doubled compared with SHA-256 leading to large gap between $c_{256}^{\cf}$ and $c_{384}^{\cf}$.
In contrast, $c_{384}^{\cf}$ and $c_{512}^{\cf}$ do not show much difference as SHA-384 and SHA-512 algorithms are identical except truncation and initial values.
The result of Sect.~\ref{subsec:CNS} is also extended to SHA-384 and reflected in Fig.~\ref{fig:security_strengths}.

Once trade-offs coefficients are obtained, we are able to draw the security strength of each algorithm in terms of required qubits as a function of Toffoli-depth.
Note that somewhere between ${\rm MAXDEPTH}=2^{64}$ and $2^{96}$, security strengths of SHA-256 (SHA-384) and AES-192 (AES-256) are reversed in order, due to their different trade-offs curve behaviors.
One minor note is that for large MAXDEPTH (for example, $2^{96}$), Proposition~\ref{prop:to-KS} does not exactly hold since the size of the domain is larger than that of the codomain in AES-192 and AES-256.
This factor is handled in a conservative way and reflected in Fig.~\ref{fig:security_strengths}.

\renewcommand{\arraystretch}{1.4}
\begin{figure}[htbp]
    \centering
    \raisebox{-0.5\totalheight}{
        \includegraphics[width=0.65\textwidth]{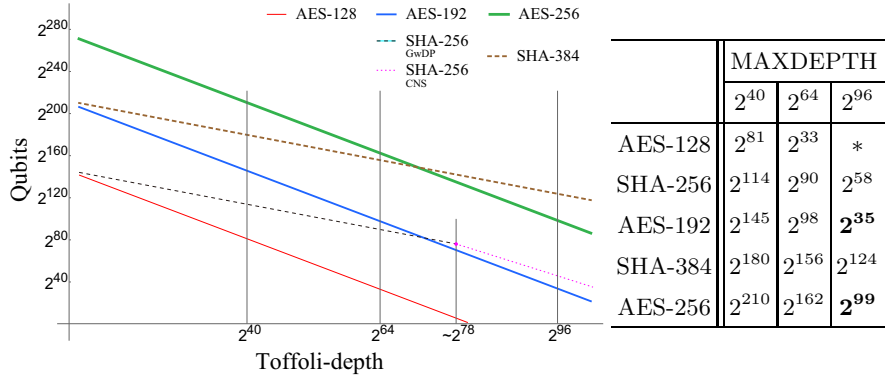}
    }
\!\!\!\!
    \small
    \begin{tabular}{>{\centering}p{1.3cm} || >{\centering}p{0.6cm} |
                  >{\centering}p{0.6cm}  | >{\centering}p{0.6cm} }  \hline
            &        \multicolumn{3}{c}{MAXDEPTH}                                \tabularnewline \cline{2-4}
    \T      &$ 2^{40}  $&$ 2^{64}  $&$ 2^{96}  $\tabularnewline \hline
    AES-128 &$ 2^{81}  $&$ 2^{33}  $&$    *    $\tabularnewline
    SHA-256 &$ 2^{114} $&$ 2^{90}  $&$ 2^{58}  $\tabularnewline
    AES-192 &$ 2^{145} $&$ 2^{98}  $&$ \mathbf{2^{35}}  $\tabularnewline
    SHA-384 &$ 2^{180} $&$ 2^{156} $&$ 2^{124} $\tabularnewline
    AES-256 &$ 2^{210} $&$ 2^{162} $&$ \mathbf{2^{99}}      $\tabularnewline
  \hline
  \end{tabular}
  \normalsize
    \caption{Security strengths of AES and SHA-2. Values in the table are approximated number of qubits required to run the respective algorithm for given MAXDEPTH.}
    \label{fig:security_strengths}
\end{figure}
\renewcommand{\arraystretch}{1.0}

Figure\;\ref{fig:security_strengths} summarizes the results which can be interpreted as another threshold to be used, for the security strength classification of proposed schemes in NIST PQC standardization process.

%
%


\section{Summary}\label{sec:con}
Instead of conventional query complexity, we have examined the time-space complexity of Grover's algorithm and its variants.
Three categories of cryptographic search problems and their characteristics are carefully considered in conjunction with probabilistic nature of quantum search algorithms.

To relate the time-space complexity with physical quantity,
we have proposed a way of quantifying the computational power of quantum computers.
Despite its simplicity, counting the number of sequential Toffoli gates reflects the reliable time complexity in estimating security levels of symmetric cryptosystems.
With simplified cost measure, one can estimate the quantum complexity of a cryptosystem concisely by counting (and focusing) relevant operations only.
It is worth noting that the above scheme is general for quantum resource estimates in symmetric cryptanalysis.

The scheme has been applied to resource estimates for AES and SHA-2.
When multiple quantum trade-offs options are given, the time-space complexity provides clear criteria to tell which is more efficient.
Based on the trade-offs observations made in AES and SHA-2, security strengths of respective systems are investigated with the MAXDEPTH assumption.

\subsubsection{Acknowledgement.}
We are grateful to Brandon Langenberg, Martin Roetteler, and Rainer Steinwandt for helpful discussion and sharing details of their previous work which has motivated us.

\newpage

\bibliographystyle{./splncs03}
\bibliography{reference}



\end{document}